\documentclass[12pt]{article}
\usepackage{makeidx}
\usepackage[combine]{ucs}
\usepackage[utf8x]{inputenc}
\usepackage{epsfig,mathrsfs}
\usepackage{lmodern,amsmath,amsthm,amssymb,graphicx,subfigure,xcolor}
\usepackage{pslatex,xspace}
\usepackage[a4paper]{geometry}%       margins and paper size
\usepackage[boxed,vlined]{algorithm2e} % algorithm
\usepackage{rotating}
\usepackage{amstext}
\usepackage{amsopn,enumitem}
\usepackage{tkz-graph}
\tikzstyle{vertex}=[circle, draw, inner sep=0pt, minimum size=6pt]
\newcommand{\vertex}{\node[vertex]}

\usepackage{tikz}
\usetikzlibrary{shapes,arrows,decorations.markings,calc}
\usepackage[pdftex,                %
    bookmarks         = true,%     % Signets
    bookmarksnumbered = true,%     % Signets numÌ©rotÌ©s
    pdfpagemode       = none,%     % Signets/vignettes fermÌ© ÌÊ l'ouverture
    pdfstartview      = FitH,%     % La page prend toute la largeur
    pdfpagelayout     = SinglePage,% Vue par page
    colorlinks        = true,%     % Liens en couleur
    urlcolor          = magenta,%  % Couleur des liens externes
    pdfborder         = {0 0 0}%   % Style de bordure : ici, pas de bordure
]{hyperref}
\newlist{exoenum}{enumerate}{3}
\setlist[exoenum,1]{label=\arabic*)}
\setlist[exoenum,2]{label=\alph*)}
\setlist[exoenum,3]{label=\roman*)}

\newcommand{\poubelle}[1]{}

\newcommand{\ra}{\rightarrow}
\newcommand{\la}{\leftarrow}
\newcommand{\N}{\hbox{I\hskip -2pt N}}

\renewcommand\geq{\geqslant}
\renewcommand\leq{\leqslant}
\renewcommand\ge{\geqslant}
\renewcommand\le{\leqslant}

\newcommand{\e}{\epsilon}

\newcommand{\rs}{\rightsquigarrow}

\DeclareMathOperator{\tri}{\chi_T}

\newenvironment{subproof}{\par\noindent {\it Proof}.\ }{\hfill$\lozenge$\par\vspace{11pt}}

\theoremstyle{plain}
\newtheorem{theorem}{Theorem}

\newtheorem{proposition}[theorem]{Proposition}
\newtheorem{claim}{Claim}[theorem]
\newtheorem{property}{Property}[theorem]

\newtheorem{corollary}[theorem]{Corollary}
\newtheorem{lemma}[theorem]{Lemma}

\theoremstyle{definition}
\newtheorem{remark}[theorem]{Remark}

\newtheorem{conjecture}[theorem]{Conjecture}

\DeclareMathOperator{\dist}{dist}

\DeclareMathOperator{\Forb}{Forb}

\DeclareMathOperator{\trans}{tt}
\DeclareMathOperator{\Or}{Or}
\DeclareMathOperator{\Reach}{Reach}

\title{\bf $\chi$-bounded families of oriented graphs\thanks{This work was supported by ANR under contract STINT ANR-13-BS02-0007.}}

\author{Pierre Aboulker\thanks{Project Coati, I3S (CNRS, UNSA) and INRIA,
Sophia Antipolis, France.}
 \and 
 J\o rgen Bang-Jensen\thanks{Department of
Mathematics and Computer Science, University of Southern Denmark,
Odense, Denmark.  Part of this work was done while this author was
visiting project COATI, Sophia Antipolis.  Hospitality and financial
support from Labex UCN@Sophia, Sophia Antipolis is gratefully
acknowledged.  The research of Bang-Jensen was also supported by the
Danish research council under grant number 1323-00178B.} 
\and 
Nicolas Bousquet\thanks{LIRIS, Ecole Centrale Lyon, France.} 
\and 
Pierre Charbit\thanks{IRIF, Universit\'e Paris Diderot, France and Project Gang , INRIA.} 
\and 
Fr\'ed\'eric Havet$^{\dag}$ 
\and 
Fr\'ed\'eric Maffray\thanks{G-SCOP, CNRS and Universit\'e Grenoble Alpes, Grenoble, France.} 
\and Jose Zamora\thanks{Departamento de Matem\'aticas, Universidad
Andres Bello, Santiago, Chile. This author was partially supported by N\'ucleo Milenio Informaci\'on y Redes ICM/FIC RC130003 and Basal program PFB-03 CMM} }
\begin{document}

\maketitle

\begin{abstract}
A famous conjecture of Gy\'arf\'as and Sumner states for any tree $T$
and integer $k$, if the chromatic number of a graph is large enough,
either the graph contains a clique of size $k$ or it contains $T$ as
an induced subgraph.  We discuss some results and open problems about
extensions of this conjecture to oriented graphs.  We conjecture that
for every oriented star $S$ and integer $k$, if the chromatic number
of a digraph is large enough, either the digraph contains a clique of
size $k$ or it contains $S$ as an induced subgraph.  As an evidence,
we prove that for any oriented star $S$, every oriented graph with
sufficiently large chromatic number contains either a transitive
tournament of order $3$ or $S$ as an induced subdigraph.  We then
study for which sets ${\cal P}$ of orientations of $P_4$ (the path on
four vertices) similar statements hold.  We establish some positive
and negative results.
\end{abstract}

\section{Introduction}
%%%%%%%%%%%

What can we say about the induced subgraphs of a graph $G$ with large
chromatic number?  Of course, one way for a graph to have large
chromatic number is to contain a large complete subgraph.  However, if
we consider graphs with large chromatic number and small clique
number, then we can ask what other subgraphs must occur.  We can avoid
any graph $H$ that contains a cycle because, as proved by
Erd\H{o}s~\cite{Erd59}, there are graphs with arbitrarily high girth
and chromatic number; but what can we say about trees?  Gy\'arf\'as
\cite{Gya75} and Sumner~\cite{Sum81} independently made the following
beautiful and difficult conjecture.

\begin{conjecture}[Gy\'arf\'as \cite{Gya75} and Sumner~\cite{Sum81}]
\label{gyarfas-sumner}
For every integer $k$ and tree $T$, there is an integer $f(k,T)$ 
such that every graph with chromatic number at least $f(k,T)$  
contains either a clique of size $k$, or an induced copy of $T$. 
\end{conjecture}

We can rephrase this conjecture, using the concept of $\chi$-bounded
graph classes introduced by Gy\'arf\'as~\cite{Gya87}.  A class of
graph ${\cal G}$ is said to be \emph{$\chi$-bounded} if there is a
function $f$ such that $\chi(G)\leq f(\omega(G))$ for every $G\in
{\cal G}$; such a function $f$ is called a \emph{$\chi$-bounding}
function.  For instance, the class of perfect graphs is $\chi$-bounded
with $f(k)=k$ as a $\chi$-bounding function.

For a graph $H$, we write $\Forb(H)$ for the class of graphs that do
not contain $H$ as an induced subgraph.  For a class of graphs ${\cal
H}$, we write $\Forb({\cal H})$ for the class of graphs that contain
no member of ${\cal H}$ as an induced subgraph.  As we have remarked,
$\Forb(H)$ is not $\chi$-bounded when $H$ contains a cycle.  The
conjecture of Gy\'arf\'as and Sumner (Conjecture~\ref{gyarfas-sumner})
asserts that $\Forb(T)$ is $\chi$-bounded for every tree $T$.  In
fact, an easy argument shows that the conjecture is equivalent to the
following one .

\begin{conjecture}\label{conj:chi-bounded}
$\Forb(H)$ is $\chi$-bounded if and only if $H$ is a forest.
\end{conjecture}

There are not so many cases solved for this conjecture, let us recall the main ones.
\begin{itemize}
\item 
Stars: Ramsey's Theorem implies easily that $\Forb(K_{1,t})$ is
$\chi$-bounded for every $t$.  
\item 
Paths: Gy\'arf\'as~\cite{Gya87} showed that $\Forb(P)$ is
$\chi$-bounded for every path $P$.  
\item 
Trees of radius 2: using the previous result, Kierstead and
Penrice~\cite{KiPe94} proved that $\Forb(T)$ is $\chi$-bounded for
every tree $T$ of radius two (generalizing an argument of Gy\'arf\'as,
Szemer\'edi and Tuza~\cite{GST80} who proved the triangle free case).
This result is proved using a result, attributed to Hajnal and R\"odl
(see \cite{KiPe94}) but apparently denied by Hajnal (see
\cite{KiRo96}), stating that $\Forb(\{T,K_{n,n}\})$ is $\chi$-bounded
for every tree $T$ and every integer $n$.  
\item 
Subdivision of stars: it is a corollary of the following topological
version of Conjecture~\ref{gyarfas-sumner} established by
Scott~\cite{Sco97}: {\it for every tree $T$ and integer $k$, there is
$g(k,T)$ such that every graph $G$ with $\chi(G) > g(k,T)$ contains
either a clique of size $k$ or an induced copy of a subdivision of
$T$.}
\end{itemize}

\smallskip

More generally, if ${\cal H}$ is a finite class of graphs, then
$\Forb({\cal H})$ is $\chi$-bounded only if ${\cal H}$ is a
forest, and Conjecture~\ref{conj:chi-bounded} states that the converse
is true.  In contrast, there are infinite classes of graphs ${\cal H}$
containing no trees that are $\chi$-bounded.  A trivial example is the
set of odd cycles, since graphs with no (induced) odd cycles are
bipartite.  Another well-known example are {\it Berge graphs} which
are the graphs with no odd holes and no odd anti-holes as induced
subgraphs.  An induced cycle of length at least 4 is a \emph{hole}.
An induced subgraph that is the complement of a hole is an
\emph{antihole}.  A hole or antihole is {\it odd} (resp.  {\it even})
if it has a odd (resp.  even) number of vertices.  The celebrated
Strong Perfect Graph Theorem~\cite{CRST06}, states that Berge graphs
are perfect graphs, i.e. graphs for which chromatic number equals
clique number.  In other words, the class of Berge graphs is
$\chi$-bounded with the identity as bounding function.
%\begin{theorem}[Chudnovsky et al.~\cite{CRST06}]\label{strong-perfect} If $G$ is a Berge graph, then $\chi(G)=\omega(G)$.
%\end{theorem}
%We say that an oriented graph is {\it perfect} if its underlying graph its perfect.
Many super-classes of the class of Berge graphs are conjectured or proved to be $\chi$-bounded.
%\begin{theorem}[Scott and Seymour~\cite{ScSe}]\label{thm:odd-hole}
%The class of odd-hole-free graphs is $\chi$-bounded. 
%\end{theorem}
In fact, Scott and Seymour~\cite{ScSe} proved that if $G$ is
odd-hole-free, then $\chi(G)\leq 2^{3^{\omega(G)}}$.  This upper bound
is certainly not tight.  Better bounds are known for small values of
$\omega(G)$.  If $\omega(G)=2$, then $G$ has no odd cycles and so is
bipartite.  If $\omega(G)=3$, then $\chi(G)\leq 4$ as shown by
Chudnovsky et al.~\cite{CRST10}.

\begin{theorem}[Chudnovsky et al.~\cite{CRST10}]\label{thm:CRST10}
Every odd-hole-free graph with clique number at most $3$ has chromatic
number at most $4$.
\end{theorem}

\bigskip

The goal of this paper is to extend some results known about
Conjecture \ref{gyarfas-sumner}.  Let $T$ be a tree for which we know
that Conjecture \ref{gyarfas-sumner} is true, and let ${\cal D}_T$ be a set of orientations
of $T$.  Then one can consider the class $\Forb({\cal D})$ of oriented
graphs that have an orientation without any induced subdigraph in
${\cal D}_T$.  Different sets ${\cal D}_T$ will define different
superclasses of $\Forb(T)$, and one can wonder which of these are
still $\chi$-bounded.  Equivalently, if one defines the chromatic
number or clique number of an oriented graph to be that of its
underlying graph, one can also talk about a $\chi$-bounded classes of
oriented graphs, and we can ask which set of oriented trees, when
forbidden as induced subdigraphs, defines $\chi$-bounded classes of
oriented graphs.  After a section establishing notations and basic
tools, we consider oriented stars (i.e. orientations of $K_{1,n}$) and
oriented paths (i.e. orientations of paths).

Before detailing those results, let us note that in this oriented
setting, if we do not demand the subdigraph to be induced, then the
problem is radically different.  Burr proved that every
$(k-1)^2$-chromatic oriented graph contains every oriented tree of
order $k$.  This was slightly improved by Addario-Berry et
al.~\cite{AHS+13} by replacing $(k-1)^2$ by $(k^2/2-k/2+1)$.  The
right bound is conjectured~\cite{Burr80} to be $(2k-2)$.

\subsection{Oriented stars}
%%%%%%%%%%%%%

 %Observe that by directional duality, $\Forb(S_{k,\ell})=\Forb(S_{\ell,k})$, and that  $\Forb(S_{k,\ell}) \subseteq \Forb(S_{k',\ell'})$ for $k\leq k'$ and $\ell\leq \ell'$. 
 We conjecture the following : 
%We first consider the case when $H$ is a star. We conjecture that for any orientation $F$ of $H$,
%$\Forb(F)$ is $\chi$-bounded (which implies that for every non-empty subset ${\cal F}$
%of $\Or(H)$, the class $\Forb({\cal F})$ is $\chi$-bounded).

\begin{conjecture}\label{conj:star}
For any oriented star $S$, $\Forb (S)$ is $\chi$-bounded.
\end{conjecture}

For every choice of positive integers $k, \ell$, we denote by $S_{k,\ell}$ the
oriented star on $k+\ell +1$ vertices where the center has in-degree
$k$ and out-degree $\ell$.  Of course by directional duality the
result for $S_{k,\ell}$ implies the result for $S_{\ell,k}$.  Also, since $\Forb(S_{k,\ell}) \subseteq \Forb(S_{k,k})$ if $k\geq \ell$, it suffices 
to prove the conjecture for $S_{k,k}$ for all values of $k$.

The cases $k=0$ and $k=\ell=1$ are not difficult and were previously
known (as mentioned in \cite{KiRo96}) but no proof was published.  As
those proofs are short and interesting, we provide them in
Subsection~\ref{subsec:connu}.

By definition of $\chi$-boundedness, Conjecture~\ref{conj:star} can be
restated as follows: for every positive integer $p$, $\Forb (\Or(K_p),
S)$ has bounded chromatic number, where $\Or(K_p)$ is the set of
orientations of $K_p$ (that is, $Or(K_p)$ is the set of all tournaments on $p$ vertices).  There are exactly two orientations of $K_3$ :
the directed cycle on three vertices $\vec{C_3}$, and the transitive
tournament on three vertices $TT_3$.  It is not difficult to show
that, for any oriented star $S$, $\Forb (\vec{C_3}, TT_3, S)$ has
bounded chromatic number.  We can even determine the exact value of
$\chi(\Forb (\vec{C_3}, TT_3, S))$
(Proposition~\ref{prop:star-triangle}).  This can be seen as the first
step ($p=3$) of Conjecture~\ref{conj:star}.  Kierstead and
R\"odl~\cite{KiRo96} proved that $\Forb (\vec{C_3}, S)$ is
$\chi$-bounded.  In Theorem~\ref{thm:TT3-star}, we prove the following
counterpart : {\it $\Forb(TT_3, S)$ has bounded chromatic number, for
every oriented star $S$.} This can be seen as the next step towards
Conjecture~\ref{conj:star}; indeed, by Theorem~\ref{thm:erdos-moser},
every orientation of $K_4$ contains $TT_3$ as an induced subdigraph,
so $\Forb(TT_3, S)\subset \Forb(\Or(K_4), S)$.  The next step would be
to prove that $\Forb (\Or(K_4), S)$ has bounded chromatic number for
every oriented star $S$.

%

%in terms of clique number to Conjecture~\ref{conj:star} because $\Or(K_3)=\{\vec{C_3}, TT_3\}$. We even determined the exact value of $\chi(\Forb (\vec{C_3}, TT_3, S))$ (Proposition~\ref{prop:star-triangle}).
%Kierstead and R\"odl~\cite{KiRo96} proved that the larger class $\Forb (\vec{C_3}, S)$ is also $\chi$-bounded.
%We prove the following counterpart.
%   
% \begin{theorem}\label{thm:TT3-star}
% For every oriented star $S$,  $\Forb(TT_3, S)$ has bounded chromatic number. 
% \end{theorem}
%This can be seen as the next step towards Conjecture~\ref{conj:star} because $\Forb(\Or(K_3),S) \subset  \Forb(TT_3, S)\subset \Forb(\Or(K_4), S)$ 
%by Theorem~\ref{thm:erdos-moser}.
%
%\medskip

\subsection{Oriented paths on four vertices}
%%%%%%%%%%%%%%%%%%%%%%

Let us denote by $P_k$ the path on $k$ vertices.  Since $P_2$ and
$P_3$ are stars, the next case for paths concerns orientations of
$P_4$.  The graphs with no induced $P_4$ are known as {\it cographs},
and it is well-known that cographs are perfect.  In particular, the
class of cographs is $\chi$-bounded (or equivalently $\Forb(\Or(P_4))$
is $\chi$-bounded).  There are four non-isomorphic orientations of
$P_4$.  They are depicted in Figure~\ref{fig:P4}.
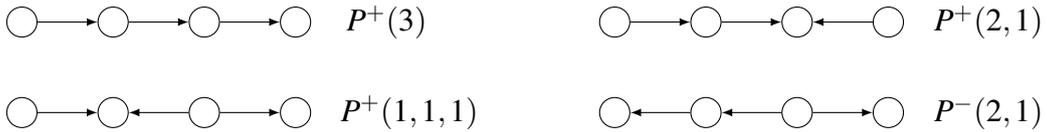
\begin{figure}[!hbtp]
\centering
 \begin{tikzpicture}[scale=0.6]
%P+(3)
\node[draw,circle] (1) at (0,0) {};
\node[draw,circle] (2) at (2,0) {};
\node[draw,circle] (3) at (4,0) {};
\node[draw,circle] (4) at (6,0) {};

\draw[->,>=latex] (1) -- (2);
\draw[->,>=latex] (2) -- (3);
\draw[->,>=latex] (3) -- (4);

\draw (8,0) node {$P^+(3)$};

\begin{scope}[xshift=13cm]
\node[draw,circle] (1) at (0,0) {};
\node[draw,circle] (2) at (2,0) {};
\node[draw,circle] (3) at (4,0) {};
\node[draw,circle] (4) at (6,0) {};

\draw[->,>=latex] (1) -- (2);
\draw[->,>=latex] (2) -- (3);
\draw[->,>=latex] (4) -- (3);

\draw (8.2,0) node {$P^+(2,1)$};
\end{scope}

\begin{scope}[yshift=-2cm, xshift=13cm]
\node[draw,circle] (1) at (0,0) {};
\node[draw,circle] (2) at (2,0) {};
\node[draw,circle] (3) at (4,0) {};
\node[draw,circle] (4) at (6,0) {};

\draw[->,>=latex] (2) -- (1);
\draw[->,>=latex] (3) -- (2);
\draw[->,>=latex] (3) -- (4);

\draw (8.2,0) node {$P^-(2,1)$};
\end{scope}

\begin{scope}[yshift=-2cm]
\node[draw,circle] (1) at (0,0) {};
\node[draw,circle] (2) at (2,0) {};
\node[draw,circle] (3) at (4,0) {};
\node[draw,circle] (4) at (6,0) {};

\draw[->,>=latex] (1) -- (2);
\draw[->,>=latex] (3) -- (2);
\draw[->,>=latex] (3) -- (4);

\draw (8.5,0) node {$P^+(1,1,1)$};
\end{scope}

\end{tikzpicture}
  \caption{The four orientations of $P_4$}
    \label{fig:P4}
\end{figure}

In Section~\ref{secPaths}, we study $\Forb({\cal P})$ when ${\cal P}$
is a set of orientations of $P_4$.  Kierstead and
Trotter~\cite{KiTr92} proved that $\Forb(P^+(3))$ is not
$\chi$-bounded by constructing $(TT_3, P^+(3))$-free oriented graphs
with arbitrary large chromatic number.  Gy\'arf\'as pointed out that
the natural orientations of the so-called \emph{shift graphs}
(\cite{ErHa76}) are in $\Forb(\vec{C}_3, TT_3,P^+(1,1,1))$ but may
have arbitrarily large chromatic number.  Consequently,
$\Forb(P^+(1,1,1))$ is not $\chi$-bounded.  See
Subsection~\ref{subsec-nochi}.

\medskip

We believe that $\{P^+(3)\}$ and $\{P^+(1,1,1)\}$ are the only
non-empty subsets ${\cal P}$ of $\Or(P_4)$ such that $\Forb({\cal P})$
is not $\chi$-bounded.

\begin{conjecture}\label{conj:P4}
Let ${\cal P}$ be a non-empty subset of $\Or(P_4)$.\\
If ${\cal P} \neq \{P^+(3)\}$ and ${\cal P}\neq \{P^+(1,1,1)\}$, then
$\Forb({\cal P})$ is $\chi$-bounded.
\end{conjecture}

\noindent We prove this conjecture in the case when $P^+(3)\in {\cal
P}$: in Corollary~\ref{cor:2P4}, we show that the classes
$\Forb(P^+(3), P^+(2,1))$, $\Forb(P^+(3), P^-(2,1))$, and
$\Forb(P^+(3), P^+(1,1,1))$ are $\chi$-bounded.  Hence, it remains
to prove Conjecture~\ref{conj:P4} for ${\cal P}\subseteq
\Forb(P^+(2,1), P^-(2,1), P^+(1,1,1))$.  Several results in this
direction have been established.  Kierstead (see~\cite{Sc85}) proved
that every $(\vec{C_3}, P^+(2,1), P^-(2,1))$-free oriented graph $D$
can be coloured with $2^{\omega(D)}-1$ colours, so in particular
$\Forb(\vec{C_3},P^+(2,1), P^-(2,1))$ is $\chi$-bounded.
Chv\'atal~\cite{Ch84} proved that acyclic $P^+(2,1)$-free oriented
graphs are perfect, so $\Forb(\vec{\cal C}, P^+(2,1))$ is
$\chi$-bounded.  Kierstead and R\"odl~\cite{KiRo96} generalized those
two results by proving (but with a larger bounding function) that
$\Forb(\vec{C_3}, P^+(2,1))$ is $\chi$-bounded.  In
Subsection~\ref{subsec:P21}, we make the first two steps towards the
$\chi$-boundedness of $\Forb(P^+(2,1))$.  We prove
$\chi(\Forb(\vec{C_3}, TT_3, P^+(2,1))=3$ and $\chi(\Forb(TT_3,
P^+(2,1))=4$.

%Again, some results were already known, we will discuss them and extensions that we have obtained.

%{\bf doit-on le dire?}

\section{Definitions, notations and useful facts} \label{secDef}
%%%%%%%%%%%%%

Let $D$ be a digraph.  If $uv$ is an arc we say that $u$ {\it
dominates} $v$ and write $u\ra v$.  Let $P=(x_1,x_2,\cdots, x_n)$ be
an oriented path.  We say that $P$ is an {\it $(x_1,x_n)$-path}.  The
vertex $x_1$ is the {\it initial vertex} of $P$ and $x_n$ its {\it
terminal vertex}.  Then we say that $P$ is a {\it directed path} or
simply a {\it dipath}, if $x_i\ra x_{i+1}$ for all $1\leq i\leq n-1$.
%A \emph{block} of $P$ is a maximal directed subpath of $P$. A path is entirely determined by the sequence $(b_1, \dots , b_p)$ of the lengths of its blocks and the sign $+$ or $-$ indicating if the first arc is forward or backward respectively. Therefore we denote by $P^+(b_1, \ldots, b_p)$ (resp. $P^-(b_1, \ldots, b_p)$) the out-path (resp. in-path) with $p$ blocks, such that the $i$th block along it has length $b_i$.  
An oriented cycle $C=(x_1,x_2 \dots x_n,x_1)$ is a {\it directed
cycle}, if $x_i\ra x_{i+1}$ for all $1\leq i\leq n$, where
$x_{n+1}=x_1$.  The directed cycle of length $k$ is denoted by
$\vec{C}_k$.

The digraph $D$ is {\it connected} % (resp.\ {\it $k$-connected})
 if its underlying graph is connected. % (resp.\ $k$-connected).
It is {\it strongly connected}, or {\it strong}, if for any two
vertices $u,v$, there is a $(u,v)$-dipath in $D$.  A strong component
$U$ is {\it initial} if all the arcs with head in $U$ have their tail
in $U$.  We denote by $\vec{\cal C}$ the class of directed cycles and
by $\cal S$ the class of strong oriented graphs.
%hence $\Forb(\vec{\cal C})$ is the set of acyclic oriented graphs.

The {\it chromatic number} (resp.  {\it clique number}) of a digraph,
denoted by $\chi(D)$ (resp.  $\omega(D)$), is the chromatic number
(resp.  clique number) of its underlying graph.  The {\it chromatic
number} of a class ${\cal D}$ of digraphs, denoted $\chi({\cal D})$,
is the smallest $k$ such that $\chi(D)\leq k$ for all $D\in {\cal D}$,
or $+\infty$ if no such $k$ exists.  If $\chi({\cal D}) \neq +\infty$,
we say that ${\cal D}$ {\it has bounded chromatic number}.  Similarly
to undirected graphs, a class of oriented graphs ${\cal D}$ is said to
be \emph{$\chi$-bounded} if there is a function $f$ such that
$\chi(D)\leq f(\omega(D))$ for every $D\in {\cal D}$ (such a function
$f$ is called a \emph{$\chi$-bounding function} for the class).

Let $F$ be a digraph and let ${\cal F}$ be a class of digraphs.  A
digraph is {\it $F$-free} (resp.  {\it ${\cal F}$-free}) if it does
not contain $F$ (resp.  any element of ${\cal F}$) as an induced
subgraph.  In this paper, we study for which classes ${\cal F}$ of
digraphs, the class of ${\cal F}$-free digraphs is $\chi$-bounded.
Observe that such an ${\cal F}$ must contain a complete (symmetric)
digraph, that is a digraph in which any two distinct vertices are
joined by two arcs in opposite direction.  Indeed, $\vec{K}_k$, the
complete digraph on $k$ vertices, has chromatic number $k$, and every
induced subdigraphs of a complete digraph is a complete digraph.

In this paper, {\bf we consider oriented graphs}, which are
$\vec{K}_2$-free digraphs.  Alternately, an {\it oriented graph} may
be defined as the orientation of a graph.  Note that an ${\cal
F}$-free oriented graph is an $({\cal F}, \vec{K}_2)$-free digraph.
We denote by $\Forb({\cal F})$ the class of ${\cal F}$-free oriented
graphs.  We are interested in determining for which class ${\cal F}$
of oriented graphs, the class $\Forb({\cal F})$ is $\chi$-bounded.  To
keep notation simple, we abbreviate $\Forb(\{F_1, \dots , F_p\})$ as
$\Forb(F_1, \dots , F_p)$, $\Forb({\cal F}_1\cup \cdots \cup {\cal
F}_p)$ in $\Forb({\cal F}_1, \dots , {\cal F}_p)$, $\Forb(\{F\}\cup
{\cal F}\})$ in $\Forb(F, {\cal F})$, and so on \dots

Let us denote by $\Or(G)$ the set of all possible orientations of a
graph $G$, and by $\Or({\cal G})$ the set of all possible orientations
of a graph in the class ${\cal G}$.  By definition a class of oriented
graphs ${\cal D}$ is $\chi$-bounded if and only if for every positive
integer $n$, ${\cal D}\cap \Forb(\Or(K_n))$ has bounded chromatic
number.  The Gy\'arf\'as-Sumner conjecture
(Conjecture~\ref{conj:chi-bounded}) can be restated as follows~:
$\Forb(\Or(H))$ is $\chi$-bounded if and only if $H$ is a forest.
However $\Forb({\cal F})$ could be $\chi$-bounded for some strict
subset ${\cal F}$ of $\Or(H)$.  More generally, for any result proving
that a class $\Forb({\cal G})$ is $\chi$-bounded, a natural question
is to ask for which subsets ${\cal F}$ of $\Or({\cal G})$, the class
$\Forb({\cal F})$ is also $\chi$-bounded.  For example, the result
mentioned in the introduction stating that $\Forb(\{T,K_{n,n}\})$ is
$\chi$-bounded for every tree $T$ and every integer $n$ has been
generalized to orientations of $K_{n,n}$ and oriented trees :
Kierstead and R\"odl~\cite{KiRo96} proved that for any positive
integer $n$ and oriented tree $T$, the class $\Forb(DK_{n,n}, T)$ is
$\chi$-bounded where $DK_{n,n}$ is the orientation of the complete
bipartite graph $K_{n,n}$ where all edges are oriented from a part to
the other.

%\medskip

%Observe that all induced subdigraphs of a transitive tournament are also transitive tournaments.  Therefore we have the following.
%\begin{proposition}\label{prop:transitive}
%Let ${\cal F}$ be a class of oriented graphs.
%If $\Forb({\cal F})$ has bounded chromatic number, then ${\cal F}$ contains a transitive tournament.
%\end{proposition}

A {\it tournament} is an orientation of a complete graph.  The unique
orientation of $K_n$ (the complete graph on $n$ vertices) with no
directed cycles is called the {\it transitive tournament} of order $n$
and is denoted $TT_n$.  Let $\trans(D)$ be the order of a largest
transitive tournament in $D$.  Observe that a class of oriented graphs
is $\chi$-bounded if and only if there is a function $g$ such that
$\chi(D)\leq g(\trans(D))$ for every $D\in {\cal D}$ thanks to the
following result due to Erd\H{o}s and Moser~\cite{ErMo64}.

\begin{theorem}[Erd\H{o}s and Moser~\cite{ErMo64}]\label{thm:erdos-moser}
For every tournament $T$, $\trans(T)\geq 1+ \lfloor \log |V(T)|
\rfloor$.
\end{theorem}

By the above observation, ${\cal D}$ is $\chi$-bounded if and only if
for every positive integer $n$, ${\cal D}\cap \Forb(TT_n)$ has bounded
chromatic number.  Also let us remark that any orientation of $K_4$
contains a $TT_3$, so ${\cal D}\cap \Forb(\Or(K_3)) \subset {\cal
D}\cap \Forb(TT_3) \subset {\cal D}\cap \Forb(\Or(K_4))$.  When we are
able to prove the result for $\Or(K_3)$-free oriented graphs and want
to extend the result to $\Or(K_4)$-free ones, an intermediate step is
therefore to prove the $TT_3$-free case.  We will do this in some
cases in Section $3$ and $4$.

\medskip

For a set or subgraph $S$ of $D$, we denote by $\Reach_D^+(S)$ (resp.
$\Reach_D^-(S)$, the set of vertices $x$ such that is a directed
out-path (reap.  directed in-path) with initial vertex in $S$ and
terminal vertex $x$.

Let $D$ be a digraph on $n$ vertices $v_1, \dots , v_n$.  A digraph
$D'$ is an {\it extension} of $D$ if $V(D')$ can be partitioned into
$(V_1, \dots , V_n)$ such that $A(D')=\{xy \mid x\in V_i, y\in V_j
\mbox{ and } v_iv_j\in A(D)\}$.  Observe that some $V_i$ may be empty.
In particular, induced subdigraphs of $D$ are extensions of $D$.

To finish this section let us state easy results that we will often
use in the proofs.  Recall that a $k$-critical graph is a graph of
chromatic number $k$ of which any strict subgraph has chromatic number
at most $k-1$.  For a digraph $D$, $\delta(D)$ denotes the minimum
degree of the underlying unoriented graph, and $\Delta^+(D)$ (resp.
$\Delta^-(D)$) denote the maximum out-degree (resp.~in-degree).

\begin{proposition}\label{prop:dkcrit}
If $D$ is a $k$-critical digraph, then $D$ is connected, $\delta(D)
\geq k-1$ and $\Delta^+(D),\Delta^-(D)\ge (k-1)/2$.
\end{proposition}

%\begin{proposition}\label{prop:distinct}
%Let $G$ be a critical graph.
%For any two non-adjacent vertices $u$ and $v$, $N(u)\setminus N(v)\neq\emptyset$.
%\end{proposition}

\begin{theorem}[Brooks~\cite{Bro41}]\label{thm:brooks}
Let $G$ be a connected graph.  If $G$ is not a complete graph or an
odd cycle, then $\chi(G)\leq \Delta(G)$.
\end{theorem}

  %\begin{theorem}[Gallai~\cite{Gal68}, Hasse~\cite{Has64}, Roy~\cite{Roy67}, Vitaver~\cite{Vit62}]\label{thm:gallairoy}
	%  Every digraph $D$ contains a directed path of order at least $\chi(D)$.
	%\end{theorem} 
% this theorem was not used

%A digraph $D$ is $m$-universal if every $m$-chromatic digraph contains a copy of $D$. 
%For a tree $T$, we define $\univ(T) = min\{k \mid$ every $k$-chromatic digraph $D$ contains a copy of $T\}$.

%\begin{theorem}[\cite{AHT07}]\label{thm:2blocks}
%	  Every path or order $n \geq 4$ with two blocks is $n$-universal.
%	\end{theorem}

%We denote by ${\cal O}$ the class of odd holes.

\section{Forbidding Oriented Stars}\label{secStars}
%%%%%%%%%%%%%%%%%%%%%%%%%%

In this section, we study the $\chi$-boundedness of $\Forb(S)$, for
$S$ an oriented star.

\subsection{Forbidding $S_{0,\ell}$ or $S_{1,1}$}\label{subsec:connu}
%%%%%%%%%%%%%%%%%%%%%%%%%
In \cite{KiRo96}, the authors state that the results in this section
were already known, but since they give no reference, and the proofs
are short, we include them here.  As written in the introduction, the
fact that $\Forb(K_{1,t})$ is $\chi$-bounded follows directly from the
following celebrated theorem due to Ramsey.

\begin{theorem}[Ramsey~\cite{Ram30}]\label{thm:ramsey}
Given any positive integers $s$ and $t$, there exists a smallest
integer $\mathbf{r}(s,t)$ such that every graph on at least
$\mathbf{r}(s,t)$ vertices contains either a clique of $s$ vertices or
a stable set of $t$ vertices.
\end{theorem}

Similarly, it can be used to show that $\Forb(S_{0,\ell})$ is
$\chi$-bounded.

\begin{theorem}
Let $\ell$ be a positive integer.  If $D \in \Forb (S_{0,\ell})$, then
$\chi(D) < 2\mathbf{r}(\omega(D), \ell)$.
\end{theorem}
\begin{proof}
Let $D$ be an $S_{0,\ell}$-free oriented graph and let $s=\omega(D)$.  If $\chi(D)\geq
2\mathbf{r}(s, \ell)$, then by Proposition~\ref{prop:dkcrit}, $D$ has
a vertex $v$ with out-degree at least $\mathbf{r}(\omega,\ell)$.  Now
the out-neighbourhood, $N^+(v)$, of $v$ contains no stable set of size $\ell$, for its union with $v$
would induce an $S_{0,\ell}$.  Therefore, by Theorem~\ref{thm:ramsey},
$N^+(v)$ contains a clique of $s$ vertices, which forms a clique of
size $s+1$ with $v$, contradicting that $\omega(D)=s$.
\end{proof}

Note that using Ramsey's Theorem is the only known way to prove that
$\Forb(K_{1,t})$ and $\Forb(S_{0,\ell})$ are $\chi$-bounded.  The
resulting bounding functions are very high and certainly very far from
being tight.

\begin{proposition}\label{prop:utile}
For the case of out-stars we have,

(i) $\chi(\Forb (TT_3,S_{0,2}))=3$.

(ii) For $\ell\geq 3$, $\chi(\Forb (TT_3,S_{0,\ell}))\leq 2\ell -2$.
\end{proposition}
\begin{proof}
(i) A digraph in $\Forb (TT_3,S_{0,2})$ has no vertex of out-degree at
least $2$.  Hence it is the converse of a functional digraph, and its
easy to see that the chromatic number is at most $3$, so $\chi(\Forb (TT_3,S_{0,2}))\leq
3$.

The directed odd cycles are in $\Forb (TT_3,S_{0,2})$ and have
chromatic number $3$.  This implies that $\chi(\Forb
(TT_3,S_{0,2}))=3$.

(ii) It suffices to prove that every critical digraph $D$ in $\Forb
(TT_3,S_{0,\ell})$ has chromatic number at most $2\ell-2$.  Observe
that for every vertex $v$, $N^+(v)$ induces a stable set because $D$
is $TT_3$-free.  Thus $d^+(v)\leq \ell-1$ since $D$ is
$S_{0,\ell}$-free.  Hence $|A(D)|\leq (\ell -1)|V(D)|$.

If $D$ contains a vertex of degree less than $2\ell-2$, then by
Proposition~\ref{prop:dkcrit}, $\chi(D)\leq 2\ell-2$.  If not, then
every vertex has degree exactly $2\ell-2$.  Moreover, $D$ is not a
tournament of order $2\ell-1$, because every such tournament contains
a $TT_3$.  Hence by Brooks' Theorem (Theorem~\ref{thm:brooks}),
$\chi(D)\leq 2\ell -2$.
\end{proof}

The case of $S_{1,1}$-free oriented graphs is also well known, as
these are perfect graphs, and therefore $\chi$-bounded.
$S_{1,1}$-free orientations of graphs are known as {\it quasi-transitive
oriented graphs}, and it is a result of Ghouila-Houri (\cite{GH62})
that a graph has a quasi-transitive orientation if and only if it has
a transitive orientation, that is an orientation both acyclic and
quasi-transitive (such graphs are commonly called {\it comparability
graphs}).  Note that if a graph has a transitive orientation, then
cliques correspond to directed paths; according to a classical
theorem, due independently to Gallai~\cite{Gal68}, Hasse~\cite{Has64},
Roy~\cite{Roy67}, and Vitaver~\cite{Vit62}, the chromatic number of a
digraph is at most the number of vertices of a directed path of
maximum length : this implies that comparability graphs are perfect.

%Anyway, we give here an elementary proof that  $S_{1,1}$-free oriented graphs are orientations of Berge graphs, and so $\chi$-bounded
%by  Theorem~\ref{strong-perfect}.
%
%
%\begin{proposition}\label{prop:P2-perfect}
%Every  $S_{1,1}$-free oriented graph is an oriented Berge graph. 
%\end{proposition}
%\begin{proof}
%Let $D\in \Forb(S_{1,1})$.
%
%Suppose for a contradiction that $D$ contains an oriented odd hole. This hole contains a directed path of length $2$, which is an induced $S_{1,1}$, a contradiction.
%
%Suppose now that it contains an oriented odd anti-hole $A$. Let $v_1, \dots , v_{2p+1}$ be its vertices with $v_i$ not linked to $v_{i+1}$ for all $1\leq i\leq 2p+1$ ($v_{2p+2}=v_1$). If $v_1$ dominates $v_3$, then $v_1$ also dominates $v_4$ for otherwise $(v_3,v_1,v_4)$ is an induced $S_{1,1}$.
%And so on by induction, $v_1$ is a source in $A$. Similarly, if $v_1$ is dominated by $v_3$, then $v_1$ is a sink.
%By symmetry, all vertices in $A$ are  sources or sinks. Hence, because the set of sources and the set of sinks are stable, $A$ is bipartite, which contradicts the fact that it is an anti-hole. 
%\end{proof}

Oriented graphs in $\Forb (\vec{C_3},TT_3,S_{1,1})$ and in $\Forb
(TT_3,S_{1,1})$ actually have a very simple structure as we show now.

\begin{theorem}
Every connected $(TT_3,S_{1,1})$-free oriented graph $D$ satisfies 
the following:
\begin{enumerate}
\item[(i)] 
If $D$ is $\vec{C_3}$-free, then $D$ is an extension of $TT_2$.
\item[(ii)] 
If $D$ contains a $\vec{C_3}$, then $D$ is an extension of $\vec{C_3}$.
\end{enumerate}
\end{theorem}

\begin{proof}
All vertices of an oriented graph $D \in \Forb
(\vec{C_3},TT_3,S_{1,1})$ are clearly either a source or a sink, which
implies $(i)$.

Now let $D\in \Forb (TT_3,S_{1,1})$.  If $D$ does contain no
$\vec{C_3}$, then it is an extension of $TT_2$ (and thus of
$\vec{C_3}$) and we are done.  So we may assume that $D$ contains a
$\vec{C_3}$.  Let $(A,B,C)$ be the partition of a maximal extension of
$\vec{C_3}$ in $D$ such that none of the sets $A$, $B$, $C$ is empty,
where all the arcs are from $A$ to $B$, from $B$ to $C$ and from $C$
to $A$.  If $A\cup B\cup C=V(D)$ we are done, so we may assume without
loss of generality that there exists adjacent vertices $a \in A$ and
$x \in D-(A \cup B \cup C)$.  By directional duality we may assume
that $a\ra x$.  For all $c \in C$, $c$ is adjacent to $x$ for
otherwise $\{c,a,x\}$ induces $S_{1,1}$ and $x\ra c$ for otherwise
$\{a,x,c\}$ induces a $TT_3$.

Let $c \in C$.  For all $a' \in A$, $a'$ and $x$ are adjacent for
otherwise $\{x,c,a'\}$ induces a $S_{1,1}$, and $a'\ra x$ for
otherwise $\{a',x,c\}$ induces a $TT_3$.  Moreover $x$ is not adjacent
to any vertex $b \in B$, otherwise $\{a,x,b\}$ induces a $TT_3$.
Hence $D\langle A\cup B \cup C \cup \{x\}\rangle$ is an extension of
$\vec{C_3}$ with partition $(A,B,C \cup \{x\})$ , contradicting
the maximality of $(A,B,C)$.
\end{proof}

\begin{corollary}
$\chi(\Forb(\vec{C}_3, TT_3, S_{1,1})) =2$ and $\chi(\Forb(TT_3,
S_{1,1})) =3$.
\end{corollary}

\subsection{Forbidding $TT_3$ and an oriented star}
%%%%%%%%%%%%%%%%%%%%%%%

The triangle-free case for stars is easy.
\begin{proposition}\label{prop:star-triangle}
Let $k$ and $\ell$ be two positive integers.  $\chi(\Forb (\vec{C_3},
TT_3,S_{k,\ell})) \leq 2k+2\ell-2$.
\end{proposition}
\begin{proof}
Let $D$ be a $(2k+2\ell-1)$-critical $(\vec{C_3}, TT_3)$-free oriented
graph.  Let $V^-$ be the set of vertices of in-degree less than $k$
and let $V^+$ be the set of vertices of out-degree less than $\ell$.
By Proposition~\ref{prop:dkcrit}, $\chi(D\langle V^-\rangle)\leq 2k-1$
and $\chi(D\langle V^+\rangle)\leq 2\ell-1$.  Consequently, $V^-\cup
V^+\neq V(D)$ for otherwise $D$ would be $(2k+2\ell-2)$-colourable.
Hence, there is a vertex $v$ with in-degree at least $k$ and
out-degree at least $\ell$.  Thus $v$ is the center of an $S_{k,\ell}$,
which is necessarily induced because $D$ is $(\vec{C_3}, TT_3)$-free.
\end{proof}

Kierstead and R\"odl~\cite{KiRo96} proved that the class $\Forb
(\vec{C_3}, S_{k,\ell})$ is $\chi$-bounded (without providing any
explicit bound).  The goal of this section is to prove the following
counterpart to that theorem.

\begin{theorem}\label{thm:TT3-star}
For every positive integers $k$ and $\ell$, the class $\Forb(TT_3,
S_{k,\ell})$ has bounded chromatic number.
\end{theorem}
 
As mentioned in the introduction this can be seen as the next step
towards Conjecture~\ref{conj:star} because
$\Forb(\Or(K_3),S_{k,\ell})) \subset \Forb(TT_3, S_{k,\ell}))\subset
\Forb(\Or(K_4), S_{k,\ell}))$.

As already mentionned, in order to prove Theorem~\ref{thm:TT3-star} it
suffices to prove the following one.
\begin{theorem}\label{thm:Skk-TT3}
For every positive integer $k$, $\Forb(TT_3,S_{k,k})$ has bounded
chromatic number.
\end{theorem}
The proof of Theorem~\ref{thm:Skk-TT3} is given in the next
subsections.

%To support this conjecture, we may observe that: 

%\begin{proposition}\label{prop:sij}
%	  $\univ(S_{i,j}) = 2i+2j-1$, for $1 \leq i \leq j$.
%	\end{proposition}

%\begin{theorem}\label{cor:stars}
%$\chi(\Forb (\vec{C_3}, TT_3,S_{i,j})) = 2i+2j-2$.
%\end{theorem}

%Next step is to prove that, for any positive integer $k$, $\Forb(TT_3, S_{k,k})$ have bounded chromatic number, which is the aim of the rest of this subsection. 

\subsubsection{Reducing to triangle-free colouring}
%%%%%%%%%%%%%%%%%%%%%%%%/

Let $D$ be a digraph.  A \emph{triangle-free colouring} is a colouring
of the vertices such that no triangle is monochromatic.  The
\emph{triangle-free chromatic number, denoted by $\tri(D)$}, of $D$ is
the minimum number of colours in a triangle-free colouring of $D$.

\begin{lemma}\label{tchi}
If $D \in \Forb(TT_3,S_{k,k})$, then $\chi(D) \le (4k-2)\cdot \tri(D)$
\end{lemma}
\begin{proof}
Let $V_1, \dots, V_{\tri(D)}$ be a triangle-free colouring of $D$.
For every $i \leq \tri(D)$, the graph $D[V_i]$ is $(\vec{C_3},TT_3,
S_{k,k})$-free and thus $\chi(D[V_i]) \le 4k-2$ by
Proposition~\ref{prop:star-triangle}.  Hence $\chi(D) \le (4k-2)\cdot
\tri(D)$.
\end{proof}

Lemma~\ref{tchi} implies that in order to Theorem~\ref{thm:Skk-TT3} 
it is sufficient to prove the following theorem.
\begin{theorem}\label{Skk-vague}
For any positive integer $k$,  $\tri( \Forb(TT_3,S_{k,k})) < +\infty$. 
\end{theorem}

We prove Theorem~\ref{Skk-vague} in Section~\ref{subsubsec:proof}, and
this will establish Theorem~\ref{thm:Skk-TT3} as well.  The proof
requires several preliminaries.  To make the proof clear and avoid
tedious calculations, we do not make any attempt to get an explicit
constant $C_k$ such that $\tri(\Forb(TT_3,S_{k,k})) < C_k$, because
our method yields a huge constant which is certainly a lot larger than
$\tri( \Forb(TT_3,S_{k,k}))$.

\subsubsection{Preliminaries}
%%%%%%%%%%%%%%%%

\paragraph{A combinatorial lemma.} 
We start with a combinatorial lemma that only serves to prove
Lemma~\ref{kp}.
\begin{lemma}\label{comb}
Let $k \in \mathbb N$ and $p \in\, ]0,1[$.  Then there is an integer
$N(k,p)$ that satisfies the following: \\
If $H=(V,E)$ is a hypergraph where all hyperedges have size at
least $p|V|$, and the intersection of any $k$ hyperedges has size
at most $k-1$, and $|V|\ge N(k,p)$, then $|E|<k/p^k$.
% Let $k \in \mathbb N$ and $p \in ]0,1[$. %Let $r= k/p^k$. 
% Let $H=(V,E)$ be a hypergraph with all hypergedges of size at least
% $p|V|$ and such that $k$ hyperedges have at most $(k-1)$ common
% vertices.
% 
% There exists an integer $N(k,p)$ such that, if $|V|\ge N(k,p)$, then
% $|E|<k/p^k$.
\end{lemma}

\begin{proof}
%Let $V$ be an $ n$-set and let $H=(V,E)$ be a hypergraph with all hypergedges of size at least $pn$ and such that $k$ hyperedges have at most $k-1$ common vertices. 

Set $|V|=n$.  We need to prove that if $n$ is sufficiently large, then
$|E|<k/p^k$.  Let $\varphi:V^k \to \mathbb N$ be the function defined
as follows: for any $k$-subset $T$ of $V$, let $\varphi(T)=|\{A\in E
\mid T \subseteq A\}|$.  Set $\Phi=\Sigma_{T\in V^k}\, \varphi(T)$.
By the hypothesis we have $\varphi(T)\le k-1$ for all $T \in V^k$, and
thus $\Phi \le {n \choose k} \cdot (k-1)$.  Since each hyperedge
contributes to at least ${pn \choose k}$ to $\Phi$, we have $\Phi \ge
|E| \cdot {pn \choose k}$.  So $|E| \cdot
{pn \choose k} \le (k-1) \cdot {n \choose k}$, and thus %
$$|E| \le (k-1) \cdot \frac{ {n \choose k}} { {pn \choose k}} \sim_{n
\to \infty}\frac{k-1}{p^k},$$
which implies the result. 
 \end{proof}

\paragraph{The constants.}
%%%%%%%%%%%%%%

All along the proofs we will use several constants; we introduce all of them here.
\begin{itemize}
\item  
$k\ge 2$ is a fixed integer (that corresponds to the forbidden
$S_{k,k}$).
\item  
$s=1-\frac{1}{2k}$.
\item  
We choose $\e \in]0,\frac{1}{2k}[$.
\item 
We choose $t \in ]s,1-\e[$ (we need $t>s$ in Lemma~\ref{ppk1} and
\ref{ppk}  and we need $t<1-\e$ in Lemma~\ref{kkk}).
\item 
$g=k/(1-t-\e)^k$ (this corresponds to the constant $k/p^k$ in
Lemma~\ref{comb} for $p=1-t-\e$).
\item 
$N_1=\max\left(N(k,1-t-\e),\frac{(1-t-\e)\cdot g}{\e}+g\right)$ where
$N$ is the function defined in Lemma~\ref{comb}.
%$N_1$ appears in Lemma~\ref{kip})
\item 
$N_2= \max(N_1, \frac{g}{t-s}+g+1)$. %($N_2$ appears in  Lemma~\ref{kkk}).
\item 
$d = \max \left(\frac{N_2}{t} + 8g, \frac{2tg}{t-s} + g\right )$.
\end{itemize}
 %and recall that for any $p \in ]0,t[$, $g>k/(1-p-\e)^k$ and $N_1>N(k,1-t-\e)$. 

\paragraph{Definitions.}
%%%%%%%%%%%%%%
Let $D$ be an oriented graph and $A$ and $B$ be two disjoint stable
sets.  The graph $D[A,B]$ is the bipartite graph with parts $A$ and
$B$.  If $D[A,B]$ is $\overline K_{k,k}$-free and all its arcs are
from $A$ to $B$, we write $A \rs B$.  Note that $A \rs B$ implies $A
\rs C$ for every $C \subseteq B$.  Let $0<\tau<1$.  By $A \to_\tau B$,
we mean:

\begin{itemize}
\item there is no arc from $B$ to $A$,
\item for every $a \in A$, we have $d^+_B(a) \ge \tau|B|$ and 
\item for every $b \in B$, we have $d^-_A(b) \ge \tau |A|$. 
\end{itemize}

If $A \rs B$  and $A\to_\tau B$, we write  $A \rs_\tau B$.

\paragraph{The tools.}
%%%%%%%%%%%%%

We now prove several lemmas that will be used in the proof.

\begin{lemma}\label{nbr}
Let $D \in \Forb(TT_3, S_{k,k})$.  Let $x \in V(D)$.  Then $N^+(x) \rs
N^-(x)$.
\end{lemma}
\begin{proof}
Since $D$ is $TT_3$-free, $N^+(x)$ and $N^-(x)$ are stable sets and
any arc between $N^+(x)$ and $N^-(x)$ has its tail in $N^+(x)$ and its
head in $N^-(x)$.  Since $D$ is $S_{k,k}$-free, $D[N^-(x),N^+(x)]$ is
$\overline K_{k,k}$-free.
\end{proof}

The next lemma roughly states that if $A$ and $B$ are two large enough
disjoint stable sets such that $A \rs B$, then up to deleting a few
vertices from $A$ and $B$ we have $A \rs_t B$.
%Recall that $g(t,\epsilon)=k/(1-t-\e)^k$. 

\begin{lemma}\label{kp}
%\textcolor{red}{Let $\epsilon \leq \frac{1}{2k}$ and $t > \epsilon$.}
Let $A,B$ be two disjoint stable sets such that $A\rs B$.  If $|A|,
|B|\geq N_1$, then there exist $A_1 \subseteq A$ and $B_1 \subseteq B$
such that:
\begin{itemize}
\item $|A_1| \ge |A|-g$, $|B_1| \ge |B|-g$ and
\item $A_1 \rs_t B_1$. 
\end{itemize}
\end{lemma}

\begin{proof}
Assume that $|A|,|B| \ge N_1$.
%Set $p'=p+\epsilon<1$.
Let 
\[
\begin{split}
A_2=&\{a \in A: d^+_B(a)<(t+\epsilon)|B|\} \text{ and } 
A_1=A\setminus A_2, \text{ and }\\ 
B_2=&\{b \in B:d^-_{A}(b)<(t+\epsilon)|A|\} \text{ and }  B_1=B\setminus B_2.
\end{split}
\]

Let us first prove that both $|A_2|$ and $|B_2|$ are at most $g$.
Consider the hypergraph $H_B=(B,E_B)$ where $E_B=\{B\setminus N^+(a)
\mid a\in A_2\}$.  We have $|E_B|=|A_2|$ and the size of each
hyperedge of $H_B$ is at least $(1-t-\epsilon)|B|$.  Since $D[A,B]$ is
$\overline K_{k,k}$-free, $k$ vertices of $A_2$ cannot have $k$ common
non-neighbours, i.e., the intersection of any $k$ hyperedges of $H_B$
is at most $(k-1)$.  Since $|B| \ge N_1 \ge N(k,1-t-\e)$,
Lemma~\ref{comb} ensures that $|A_2|=|E_B| \le
\frac{k}{(1-t-\epsilon)^k}=g$.  Thus $|A_1|\ge |A|-g$.  Similarly
$|B_2| \le g$ and so $|B_1| \ge |B|-g$.

Since $A \rs B$, we have $A_1 \rs B_1$. Thus it remains to prove that $A_1 \to_t B_1$. 
%Set $g=g(p)$ to make the writing lighter.
Since $d_{B}^+(a)\ge (t+\epsilon)|B|$ for every $a \in A_1$, we have:
\[
\begin{split}
d^+_{B_1}(a) &\ge (t+\epsilon)|B|-|B_2|\\
&\ge t \cdot |B_1|+\epsilon|B_1|- (1-t-\epsilon) |B_2|) \hspace{20pt}\text{ (because }|B|=|B_1|+|B_2| \text{)}\\
\end{split}
\]
Now $|B_2|\le g$ and by definition of $N_1$, we have: 
$|B_1| \ge |B|-g \ge  N_1 - g\ge  \frac{(1-t-\epsilon)\cdot g}{\epsilon}. $ So $\epsilon|B_1| \geq (1-t-\epsilon)|B_2|$.  
Consequently, $d^+_{B_1}(a) \geq t \cdot |B_1|$. 

Similarly, we obtain $d^-_{A_1}(b) \geq t\cdot |A_1|$ for all $b\in
B_1$, which completes the proof.
\end{proof}

%\begin{lemma}\label{nbrp}
%Let $D \in \Forb(TT_3, S_{k,k})$  and let $x \in V(D)$.   Let $A \subseteq N^+(x)$ and $B \subseteq N^-(x)$.
%For all $p \in ]0,1-\e[$, if $|A|,|B| \ge N_1(p)$, then there exists $A_1 \subseteq A$ and $B_1 \subseteq B$ such that 
%\begin{itemize}
%\item $|A_1| \ge |A|-g(p)$,
%\item $|B_1| \ge |B|-g(p)$,
%\item $A_1 \rs_p B_1$
%\end{itemize}
%\end{lemma}
%\begin{proof}
%By Lemma~\ref{nbrk}, $A \rs B$.
%So the result follows  by applying Lemma~\ref{kp} to $D[A,B]$.
%\end{proof}

% Recall that $s=1-1/2k$.
% Next Lemma roughly says tha-t if $p$ is sufficiently near from $1$, then $A \to_p B \to_p C$ implies $C \rs A$. 
\begin{lemma}\label{ppk1}
Let $\tau \in ]s,1[$ and $D \in \Forb(TT_3,S_{k,k})$.  Let $A,B,C$ be
three disjoint stable sets of $D$.  If for every $a \in A$, $d^+_B(a)
\ge \tau |B|$ and for every $c \in C$, $d^-_B(c) \ge \tau |B|$, then
$C \rs A$.
\end{lemma}

\begin{proof}
Let us first prove that there is no arc from $A$ to $C$.  Let $a \in
A$ and $c \in C$.  Since $s > \frac 12$, we have $d^+_B(a) \ge \tau
|B|>\frac{1}{2}|B|$ and $d^-_B(c) \ge \tau |B| > \frac{1}{2} |B|$.  So
there exists $b \in B$ such that $b \in N^+(a) \cap N^-(c)$, hence
$ac$ is not an arc otherwise $\{a,b,c\}$ would induce a $TT_3$, a
contradiction.

It remains to prove that $D[C,A]$ is $\overline K_{k,k}$-free.  Assume
for contradiction that there exist $A_k=\{a_1, \dots, a_k\} \subseteq
A$ and $C_k=\{c_1, \dots, c_k\} \subseteq C$ such that there is no arc
between $A_k$ and $C_k$.  For each $a_i \in A_k$, at most
$(1-\tau)|B|$ vertices in $B$ are not in $N^+(a_i)$.  Similarly for
each $c_i \in C_k$, at most $(1-\tau)|B|$ vertices in $B$ are not in
$N^-(c_i)$.  Thus the size of $X:= \bigcap_{1\leq i \leq k} N^+(a_i) \
\cap \ \bigcap_{1\leq i \leq k} N^-(c_i) \cap B$ is at least $(1-2k (1-\tau))
|B| > (1 - 2k (1-s)) |B| = 0$.  Since $X$ is non-empty, it contains a
vertex $x$.  The set $\{x\}\cup A_k \cup C_k$ induces $S_{k,k}$, a
contradiction.
\end{proof}

All along this section, we apply this lemma with stronger assumptions.

\begin{corollary}\label{ppk}
Let $D \in \Forb(TT_3,S_{k,k})$ and  $\tau \in ]s,1[$. 
Let $A,B,C$ be three disjoint stable sets of $D$. 
If  $A \to_\tau B \to_\tau C$, then $C \rs A$.
\end{corollary}

The next lemma roughly ensures that if $A,B,C$ are three large enough
stable sets such that $A \rs B \rs C$, then, up to deleting a few
vertices from $A$ and $C$, we have $C \rs A$.

\begin{lemma}\label{kkk}
%\textcolor{red}{Let $\epsilon \leq \frac{1}{2k}$ and $t > ???$.}
Let $D \in \Forb(TT_3,S_{k,k})$ and let $A,B,C$ be three disjoint
stable sets of $D$.  If $A \rs B \rs C$ and $|A|,|B|,|C| \ge N_2$,
then there exist $A_1 \subseteq A$, and $C_1 \subseteq C$ such that:
\begin{itemize}
\item $|A_1| \ge |A|-g$,  $|C_1| \ge |C| -g$ and 
\item $C_1 \rs A_1$
\end{itemize}
\end{lemma}
\begin{proof}
The proof consists in combining Lemmas~\ref{kp} and \ref{ppk1}.  Since
$A \rs B$ and $|A|,|B| \ge N_1$, Lemma~\ref{kp} ensures that there
exist $A_1 \subseteq A$ and $B_1 \subseteq B$ such that $|A_1| \ge
|A|-g$, $|B_1| \ge |B|-g$ and $A_1 \rs_t B_1$.  Similarly, since $B\rs
C$ and $|B|,|C|>N_1$, there exist $B_2 \subseteq B$ and $C_1 \subseteq
C$ such that $|B_2| \ge |B|-g$, $|C_1| \ge |C|-g$ and $B_2 \rs_t C_1$.
\\
Set $B_3=B_1 \cap B_2$ and observe that $|B_3| \ge |B|-2g$.  Note
moreover that both $B_1 \setminus B_3$ and $B_2 \setminus B_3$ have
size at most $g$.  For all $a \in A_1$, since $A_1 \to_t B_1$ and $B_2
\to_t C_1$, we have:
\begin{equation*}
 %\begin{align}
  d^+_{B_3}(a) \ge t|B_1|-g \ge \left(t-\frac{g}{|B_1|}\right)|B_3|  \hspace{15pt} \text{ and } \hspace{15pt}   d^-_{B_3}(c) \ge t|B_2|-g \ge \left(t-\frac{g}{|B_2|}\right)|B_3|
 %\end{align}
\end{equation*}
By Lemma~\ref{ppk1}, it is sufficient to prove that $t-g/|B_i|> s$ for $i=1,2$, which is satisfied because $|B_i| \ge |B|-g \ge N_2-g > \frac{g}{t-s}$.
\end{proof}

%To simplify calculations, Lemma~\ref{tchi} is used with the function $4k \chi(D)$ in the remaining of the proof.
A digraph $D$ is \emph{$c$-triangle-free-critical} if $\tri(D)=c$ and
for all $x \in V(D)$, $\tri(D-\{x\}))<c$.

\begin{lemma}\label{tcritical}
Let $D \in \Forb(TT_3)$ be a $c$-triangle-free-critical digraph.  Then
for all $x \in V(D)$, $d^+(x)\ge c-1$ and $d^-(x) \ge c -1$.
\end{lemma}
\begin{proof}
Let $x \in V(D)$.  Let $\pi$ be a triangle-free colouring of $D-x$
using $c-1$ colours.  Since $\tri(D)=c$, we cannot extend $\pi$ to $D$
using a colour in $\{1,\ldots,c-1\}$.  Let $i \leq c-1$.  Since $x$
cannot be coloured with $i$, the vertex $x$ is adjacent to two
vertices $u_i$ and $v_i$ coloured $i$ and such that $(u_i,v_i)$ is an
arc.  Since $D$ is $TT_3$-free, necessarily, $v_i\ra x$ and $x\ra
u_i$.  Now all the $u_i$ (resp.  $v_i$) are distinct because they are
coloured with distinct colours, so $d^+(x)\geq c-1$ and $d^-(x)\geq c-1$.
\end{proof}

\subsubsection{Proof of
Theorem~\ref{Skk-vague}}\label{subsubsec:proof}
%%%%%%%%%%%%%%%%%%%

We are now able to prove Theorem~\ref{Skk-vague}.  In fact, we prove
the following theorem.

\begin{theorem}\label{Skk}
$\tri( \Forb(TT_3,S_{k,k}))\le  2 d$. 
\end{theorem}

\begin{proof}
%Let $d=\max(\frac{N_2}{t}+\textcolor{red}{7}g +1 ,  N_2+\textcolor{red}{6}g, \frac{2tg}{t-s}+g)$.
We consider a minimal counter-example, that is, a digraph $D \in
\Forb(TT_3,S_{k,k})$ which is $(2d+1)$-triangle-free-critical.  By
Lemma~\ref{tcritical}, every vertex of $D$ has in- and out-degree at
least $2d$.

Let $A,B$ be two disjoint stable sets, each of size at least
$\frac{N_2}{t}$, such that $A \rs{_t} B$ and maximizing $|A|+|B|$.  We
have:

\begin{equation}\label{A+B}
|A|+|B| \ge 4d -2g.
\end{equation}

Such sets exist. Indeed let $x$ be a vertex.  By Lemma~\ref{nbr}, we have $N^+(x) \rs N^-(x)$.
Since $2d\geq N_1$, Lemma~\ref{kp} ensures that there exists a $U
\subseteq N^-(x)$ and $V \subseteq N^+(x)$ both of size at least $2d-g$
such that $U \rs_t V$. Moreover both $U$ and $V$ are stable sets since $D$ is $TT_3$-free.
So $U$ and $V$ satisfies the conditions.

\begin{claim}\label{claim:big}
There exists $x\in A\cup B$ such that $d^+_{D-(A \cup B)}(x)\ge d$ and
$d^-_{D-(A \cup B)}(x) \ge d$.
\end{claim}

\begin{subproof} 
Assume that no vertex $x$ in $A \cup B$ satisfies $d^+_{D-(A \cup
B)}(x)\ge d$ and $d^-_{D-(A \cup B)}(x) \ge d$.  Let $\pi$ be a
triangle-free $2d$-colouring of $D-(A \cup B)$, which exists by
minimality of $D$.  For every $a \in A$, a colour $c_a$ in $\{1,
\dots, d\}$ does not appear in $\pi(N^+(a))$ or in $\pi(N^-(a))$.
Similarly, for every $b \in B$, a colour $c_b$ in $\{d+1, \dots, 2d\}$
does not appear in $\pi(N^+(b))$ or in $\pi(N^-(b))$.  Let $\pi'$ be
the colouring of $D$ where $\pi'$ agrees with $\pi$ on $D - (A \cup
B)$, and $\pi'(a)=c_a$ for every $a \in A$ and $\pi'(b)=c_b$ for every
$b \in B$.  Since $D$ is $(2d+1)$-triangle-free-critical, there is a
monochromatic oriented triangle $xyz$.  As $\pi$ is a triangle-free
colouring, at least one vertex of the triangle, say $x$, is in $A \cup
B$.  By directional duality, we may assume that $x\in A$.  Since the
colours used to colour vertices of $A$ and the colours used to colour
vertices of $B$ are disjoint, $y$ and $z$ are not in $B$.  Moreover,
since $A$ is a stable set, $y$ and $z$ are in $D-(A \cup B)$.  Thus
there is an in-neighbour and an out-neighbour of $x$ coloured with
$c_x$, a contradiction with the definition of $c_x$.
\end{subproof}

\medskip

We distinguish three cases.

\medskip

\noindent\textbf{Case 1:} $|A|\ge \frac{N_2}{t}+ 4g$ and there exists
$a \in A$ such that $d^+_{D-(A \cup B)}(a)\ge d$ and $d^-_{D-(A \cup
B)}(a)\ge d$.  \smallskip

\begin{figure}
\centering
\definecolor{ffffff}{rgb}{1.,1.,1.}
\begin{tikzpicture}[line cap=round,line join=round,>=triangle 45,x=1.0cm,y=1.0cm]
%\clip(0.82,-4.16) rectangle (23.52,6.98);
%\fill(2.,0.56) -- (2.,3.56) -- (3.,3.56) -- (3.,0.56) -- cycle;
%\fill(5.,0.) -- (5.,1.36) -- (6.,1.36) -- (6.,0.) -- cycle;
%\fill(3.44,4.72) -- (3.44,6.36) -- (4.44,6.36) -- (4.44,4.72) -- cycle;
%\fill(5.,1.62) -- (5.,2.98) -- (6.,2.98) -- (6.,1.62) -- cycle;
%\fill(4.98,4.32) -- (4.98,3.32) -- (5.98,3.32) -- (5.98,4.32) -- cycle;
\draw (2.,0.56)-- (2.,3.56);
\draw (2.,3.56)-- (3.,3.56);
\draw (3.,3.56)-- (3.,0.56);
\draw (3.,0.56)-- (2.,0.56);
\draw (5.,0.)-- (5.,1.36);
\draw (5.,1.36)-- (6.,1.36);
\draw (6.,1.36)-- (6.,0.);
\draw (6.,0.)-- (5.,0.);
\draw (3.44,4.72)-- (3.44,6.36);
\draw (3.44,6.36)-- (4.44,6.36);
\draw (4.44,6.36)-- (4.44,4.72);
\draw (4.44,4.72)-- (3.44,4.72);
\draw (5.,1.62)-- (5.,2.98);
\draw (5.,2.98)-- (6.,2.98);
\draw (6.,2.98)-- (6.,1.62);
\draw (6.,1.62)-- (5.,1.62);
\draw (4.98,4.32)-- (4.98,3.32);
\draw (4.98,3.32)-- (5.98,3.32);
\draw (5.98,3.32)-- (5.98,4.32);
\draw (5.98,4.32)-- (4.98,4.32);
\draw (2.5,2.88)-- (4.64,3.84);
\draw (4.64,3.84)-- (4.26,3.88);
\draw (4.64,3.84)-- (4.48,3.52);
\draw (2.5,2.88)-- (4.62,2.5);
\draw (4.62,2.5)-- (4.4,2.74);
\draw (4.62,2.5)-- (4.34,2.36);
\draw (3.18,1.36)-- (4.58,1.36);
\draw (4.58,1.36)-- (4.3,1.62);
\draw (4.58,1.36)-- (4.3,1.06);
\draw (2.5,2.88)-- (2.36,3.38);
\draw (2.5,2.88)-- (2.78,3.32);
\draw (6.16,4.24) node[anchor=north west] {$B_a^{ext}$};
\draw (6.16,2.8) node[anchor=north west] {$B_a^{int}$};
\draw (6.16,1.1) node[anchor=north west] {$B'$};
\draw (7,3.52) node[anchor=north west] {$B_a=B_a^{int} \cup B_a^{ext}$};
\draw (7,2.04) node[anchor=north west] {$B=B' \cup B_a$};
\draw (2.5,2.88)-- (3.18,5.22);
\draw (4.5,6.04) node[anchor=north west] {$C$};
\draw (1.44,2.5) node[anchor=north west] {$A$};
\draw (2.5 ,2.9) node {$\bullet$};
\draw[color=black] (2.27,2.71) node {$a$};
%\draw [color=black] (3.18,5.22) circle (1.5pt);

\end{tikzpicture}
  \caption{The situation in Case 1.}
  \label{figMoche}

%\caption{Illustration of the proof of Theorem~\ref{Skk}.
%\vspace{-100pt}
\end{figure}
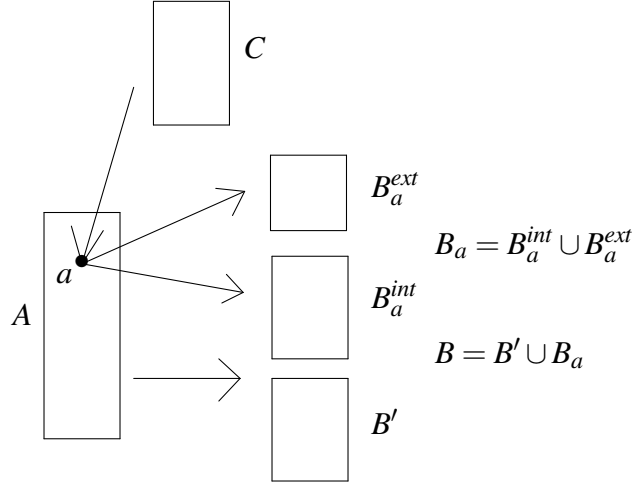

Set $C=N^-(a)$, $B_a=N^+(a)$, $B^{int}_a=N^+(a)\cap B$,
$B^{ext}_a=B_a\setminus B^{int}_a$ and $B'=B\setminus B^{int}_a$ (see
Figure~\ref{figMoche} for a rough picture of the situation).  Note
that by assumption the sizes of $B^{ext}_a$ and $C$ are both at least
$d \geq N_2+4g$.  Because $A \rs_t B$ and $B$ has size at least
$\frac{N_2}{t}$, we also have
$$ |B^{int}_a| \ge N_2.$$

Since $A \rs B$, we have $A \rs B^{int}_a$.  Moreover, $B^{int}_a
\subseteq N^+(a)$ and $C \subseteq N^-(a)$, so $B^{int}_a \rs C$ by
Lemma~\ref{nbr}.  Hence, $A \rs B^{int}_a \rs C$.  All of $A$,
$B^{int}_a$ and $C$ have size at least $N_2$, so, by Lemma~\ref{kkk},
there exist $A_1 \subseteq A$ and $C_1 \subseteq C$ such that $|A_1|
\ge |A|-g$, $|C_1| \ge |C|-g$ and $C_1 \rs A_1$.

\begin{claim}\label{claim:B'}
\[ |B'| \ge |B^{ext}_a|- 5g\geq d-5g \geq N_2+g.\]
\end{claim}
\begin{subproof}
By Lemma~\ref{nbr}, $B_a \rs C_1$.  Hence we have $B_a \rs C_1 \rs
A_1$ and $|B_a|, |C_1|, |A_1| \ge N_2 +g$.  Applying first
Lemma~\ref{kkk} and then Lemma~\ref{kp}, we obtain the existence of
$A_2 \subseteq A_1$ and $B_a^1 \subseteq B_a$ of respective size at
least $|A_1|-2g $ and $| B_a|-2g$ such that $A_2 \rs_t B_a^1$.

Now, we have $A_2 \rs_t B_a^1$, $|A_2|\ge |A| -3g \ge \frac{N_2}{t}$
and $|B_a^1| \ge |B_a|-2g \geq 2d-2g > \frac{N_2}{t}$.
%\textcolor{red}{Former useless inequality?  $>|B^{int}_a|+|B^{ext}_a| -g \ge td + d-g >\frac{N_2}{t}$}. 
The maximality of $|A|+| B|$ ensures:
\begin{eqnarray*}
|A|+|B| & \geq & |A_2| + |B_a^1|  \\
  |A| + |B_a^{int}| + |B'| & \geq & |A|-3g + |B_a^{int}| + |B_a^{ext}| -2g \\
 |B'| \geq |B_a^{ext}| -5g & \geq & d-5g 
\end{eqnarray*}
\end{subproof}

We shall now prove the existence of $A^*\subseteq A$ and $B^*\subseteq
B' \cup B^{int}_a \cup B^{ext}_a$ each of size at least $N_2/t$ such
that $A^*\rs_t B^*$ and $|A^*|+| B^*| > |A|+| B|$, which
contradicts the maximality of $|A|+| B|$.  The proof is organized as
follows: first using Lemmas~\ref{nbr}, \ref{kp} (several times),
\ref{ppk1} and \ref{kkk}, we show that almost all vertices of $B'$
have many out-neighbours in $C$.  Then we show the same for almost all
vertices in $B_a$.  Using this degree assumption and Lemma~\ref{kp},
we establish the existence of large sets $A_3\subseteq A$ and
$B_3\subseteq B' \cup B_a$ such that $A_3 \rs B_3$.  This main fact,
combined with few other calculations lead to the existence of the
above-mentioned sets $A^*$ and $B^*$.

\medskip

Since $A_1 \rs B$ and $B' \subseteq B$, we have $A_1 \rs B'$.  Thus
$C_1 \rs A_1 \rs B'$ and $|C_1|,|A_1|,|B'| \ge N_2+g$.  So
Lemma~\ref{kkk} ensures that there exist $C_2 \subseteq C_1$ and $B'_1
\subseteq B'$ such that $|C_2| \ge |C_1|-g$, $|B'_1| \ge |B'|-g$ and
$B_1'\rs C_2$.  Now, since $B_1' \rs C_2$ and $|B'_1|,|C_2| \ge N_2$,
by Lemma~\ref{kp}, there exist $B_2' \subseteq B_1'$ and $C_3
\subseteq C_2$ such that $|B_2'| \ge |B'_1|-g$ and $|C_3| \ge |C_2|-g$
such that $B_2' \rs_t C_3$.  So, for all $b \in B'_2$, we have:
\begin{equation}\label{deg1}
 d^+_{C_1}(b) \ge t \cdot |C_3| \ge t \cdot (|C_1|-2g) =
 \left(t-\frac{2tg}{|C_1|}\right)\cdot |C_1|.
\end{equation}

Lemma~\ref{nbr} ensures that $B_a \rs C_1$.  Moreover $|B_a|,|C_1| \ge
N_2$, so by Lemma~\ref{kp}, there exist $B_a^2\subseteq B_a$ and $C_4
\subseteq C_1$ such that $|B_a^2| \ge |B_a|-g$, $|C_4| \ge |C_1|-g$
and $B_a^2\textcolor{red}{\rs}_t C_4$.  So, for all $b \in B_a^2$, we
have:
 \begin{equation}\label{deg2}
 d^+_{C_1}(b) \ge t \cdot |C_4| \ge t \cdot (|C_1|-g) = \left(t-\frac{tg}{|C_1|}\right) \cdot |C_1|.
 \end{equation}

Since $C_1 \rs A_1$ and $|C_1|,|A_1| \ge N_2$, Lemma~\ref{kp} ensures
the existence of $A_3 \subseteq A_1$ and $C_5 \subseteq C_1$ such that
$|A_3| \ge |A_1|-g$, $|C_5| \ge |C_1|-g$ and $C_5
\rs\textcolor{red}{_t} A_3$.  So, for all $a \in A_3$, we have:
\begin{equation}\label{deg3}
d^-_{C_1}(a) \ge t \cdot |C_5|\ge t \cdot (|C_1-g) \ge \left(t-\frac{tg}{|C_1|}\right) \cdot |C_1|.
\end{equation}

Set $p=(t-\frac{tg}{|C_1|})$.  Then (\ref{deg1}), (\ref{deg2}) and
(\ref{deg3}) ensures that for all $b \in B'_2 \cup B_a^2$,
$d^+_{C_1}(b) \ge p|C_1|$ and for all $a \in A_3$, $d^-_{C_1}(a) \ge
p|C_1|$.  Moreover, since $|C_1| >d-g \ge \frac{2tg}{(t-s)}$ and since $t>s$, we have
$p>s$.  Thus Lemma~\ref{ppk1} yields %
$$A_3 \rs B'_2 \cup B_a^2.$$

Let us apply Lemma~\ref{kp} one last time.  Indeed, both $A_3$ and
$B'_2 \cup B_a^2$ have size at least $N_1$.  Thus, there exist $A^*
\subseteq A_3$ and $B^* \subseteq B'_2 \cup B_a^2$ of size
respectively at least $|A_3|-g$ and $| B'_2 \cup B_a^2|-g$ such that
$A^* \rs_t B^*$.

Observe that $|A^*|\geq |A_3|-g=|A_1|-2g=|A|-3g\geq \frac{N_2}{t}$.
and $|B_2'|=|B'|-2g$.  Moreover $|B^*|\geq |B'_2| -g\geq |B'_1|-2g\geq
|B'|-3g \geq d-8g$ by Claim~\ref{claim:B'}.  Since $d\geq
\frac{N_2}{t}+8g$, we have $|B^*| \geq \frac{N_2}{t}$.

Furthermore the following inequalities are satisfied:
\begin{align*}
 |A^*| + |B^*| &\geq |A_3|+|B'_2| +|B_a^2| -2g \\
	      & \ge |A| + |B_a^{int}| + |B_a^{ext}| + |B'| - 7 g \\
	      & \ge |A|+|B|+|B_a^{ext}|-7g\\
	      &>|A|+|B|.
\end{align*}
The first inequality is due to the last extraction.  The second comes
from $|A_3|\geq |A_1|-g\geq |A|-2g$ and $|B_a^2|\geq
|B_a|-g=|B_a^{int}| + |B_a^{ext}| -g $ and $|B_2'|\geq |B'|-2g$.
Finally the last inequality comes from the fact that $B^{ext}$ has
size at least $d$ which is greater than $7g$ by definition.

Thus $A^* \rs_t B^*$, $|A^*| + |B^*| > |A|+|B|$ and both $A^*$ and
$B^*$ have size at least $\frac{N_2}{t}$, a contradiction to the
maximality of $A \rs_t B$.  \medskip
 
\noindent\textbf{Case 2:} $|B| \ge \frac{N_2}{t} +4g$, and there
exists $b \in B$ such that $d^+_{D-(A \cup B)}(b)\ge d$ and $d^-_{D-(A
\cup B)}(b)\ge d$.  \smallskip
 
This case is analogous to Case 1 by directional duality.

 \medskip
 
\noindent\textbf{Case 3:} The remaining case.\smallskip
 
Claim~\ref{claim:big} ensures that there is a vertex $x$ in $A \cup B$
with in- and out-degree at least $d$.  Assume that
$x \in A$.  Since $|A| < \frac{N_2}{t} + 4g$ by Case 1 and $|A|+|B| \ge
4d -2g$ by Equation~(\ref{A+B}), we have $|B| \ge \frac{N_2}{t} +4g$.
So Case 2 ensures that no vertex $b$ of $B$ has in and out-degree at
least $d$ in the complement of $A \cup B$.

%  Hence we may assume that if $|A| \ge \frac{N_2}{t} +2g$, then $A$ does not contain a vertex with at least $d$ in and out neighbours outside $A \cup B$. Similarly, if $|B| \ge \frac{N_2}{t} +2g$, then $B$ does not contain a vertex with at least $d$ in and out neighbours outside $A \cup B$.  
%  So, by \ref{degd}, we may assume without loss of generality that $|A|<\frac{N_2}{t} +2g$.
Let $b \in B$.  Thus $b$ has in-degree at most $d+ \frac{N_2}{t}
+4g-1$ ($b$ can be incident to the vertices of $A$ plus less than $d$
vertices in $V \setminus (A \cup B)$) or $b$ has out-degree at most
$d$ (there is no arc from $B$ to $A$).  But $d+ \frac{N_2}{t} +4g-1\le
2d-1$, which contradicts Lemma~\ref{tcritical}.

The case where $x \in B$ is obtained similarly by switching Cases 1 and 2 in the proof.
\end{proof}

\section{Forbidding Oriented Paths}\label{secPaths}
%%%%%%%%%%%%%%%%%%%%%%%%%%

\subsection{Forbidding $P^+(3)$ or  $P^+(1,1,1)$}\label{subsec-nochi}
%%%%%%%%%%%%%%%%%%%%%%%%%

Kierstead and Trotter~\cite{KiTr92} proved that $\Forb(P^+(3))$ is not
$\chi$-bounded.  In fact, they show that an analogue of Zykov's
construction of triangle-free graphs with arbitrarily large chromatic
number yields acyclic $(TT_3, P^+(3))$-free oriented graphs with
arbitrary large chromatic number.  Interestingly, a result of
Galeana-S\'anchez et al.~\cite{GGU10} implies that
$\chi(\Forb(\vec{C}_3, TT_3, P^+(3))\cap {\cal S})=2$.
Galeana-S\'anchez et al.~\cite{GGU10} studied {\it
$3$-quasi-transitive digraphs}, which are digraphs in which for every
directed walk $(u,v,w,z)$ either $u$ and $z$ are adjacent or $u=z$.
In particular, every $(\vec{C}_3, TT_3, P^+(3))$-free oriented graph
is $3$-quasi-transitive.  They characterized the strong $3$-quasi-transitive
digraphs.  They showed that every such graph is either semicomplete,
or semicomplete bipartite, or in the set ${\cal F}$ of oriented graphs
$D$ that have three vertices $\{v_1, v_2, v_3\}$ such that
$A(D)=\{v_1v_2, v_2v_3, v_3v_1\} \cup \bigcup_{u\in V(D)\setminus
\{v_1,v_2,v_3\}} \{v_1u, uv_2\}$.  Recall that a digraph $D$ is {\it
semicomplete} if for any two vertices $u, v\in V(D)$ at least one of
the two arcs $uv$ and $vu$ is in $A(D)$, and that it is {\it
semicomplete bipartite} if there is a bipartition $(A,B)$ of $V(D)$
such that if for any $a\in A$ and $b\in B$, at least one of the two
arcs $ab$ and $ba$ is in $A(D)$.  Since semicomplete digraphs and
members of ${\cal F}$ are not $(\vec{C}_3, TT_3)$-free, strong
$(\vec{C}_3, TT_3, P^+(3))$-free oriented graphs are bipartite
tournaments and consequently have chromatic number at most $2$.  On
the other hand, $\Forb( P^+(3)) \cap {\cal S}$ is not $\chi$-bounded.
Indeed, adding to every acyclic $(TT_3, P^+(3))$-free oriented graph
$D$ a vertex $x$ which dominates all sources of $D$ and is dominated
by all other vertices, we obtain a strong $(\Or(K_4), P^+(3))$-free
oriented graph $D'$ with chromatic number $\chi(D)+1$; since $\chi
(\Forb(\vec{\cal C}, TT_3, P^+(3))) = +\infty$, we get $\chi
(\Forb(\Or(K_4), P^+(3)) \cap {\cal S}) = +\infty$.\\

The {\it shift graph} $Sh_k(n)$, introduced by Erd\H{o}s and
Hajnal~\cite{ErHa76}, is the graph whose vertices are the $k$-element
subsets of $\{1, \dots, n\}$ and two vertices $a=\{a_1, \dots, a_k\}$
and $b=\{b_1, \dots, b_k\}$ are adjacent iff $a_1<a_2=b_2<a_3=b_3<
\dots <a_{k-1}=b_{k-1}<b_k$.  Gy\'arf\'as pointed out that the natural
orientations of shift graphs are in $\Forb(\vec{C}_3,
TT_3,P^+(1,1,1))$ but may have arbitrarily large chromatic number.
Consequently, $\Forb(P^+(1,1,1))$ is not $\chi$-bounded.  Another way
of seeing this is to note that every line oriented graph (i.e. an
oriented graph which is a line digraph) is both $TT_3$-free and
$P^+(1,1,1)$-free and that the line oriented graph of an acyclic
oriented graph is also acyclic.  Now, since it is well known that the
chromatic number of the line digraph of $D$ is at least
$\log(\chi(D))$, this implies that the line oriented graphs of $TT_n$
form a family of oriented graphs in $\Forb(\vec{C}_3,
TT_3,P^+(1,1,1))$ with arbitrarily large chromatic number (which is
consistent with Gy\'arf\'as's remark since natural orientations of
shift graphs are in fact line oriented graphs).  It can be deduced
from Corollary 4.5.2 in \cite{BaGu08} that in fact the class of line
oriented graphs is exactly $\Forb (TT_3, P^+(1,1,1), C(3,1), C(2,2))$,
where $C(3,1)$ (resp.  $C(2,2)$) is the oriented cycle $(a_1, a_2,
a_3, a_4, a_1)$ such that $a_1\ra a_2\ra a_3 \ra a_4 \la a_1$ (resp.
$a_1\ra a_2\ra a_3 \la a_4 \la a_1$).  It follows that
$\Forb(\vec{\cal C}, TT_3, P+(1,1,1), C(3,1), C(2,2))$ has unbounded
chromatic number.

\subsection{Forbidding $P^+(2,1)$}\label{subsec:P21}
%%%%%%%%%%%%%%%%%%

\begin{theorem}\label{thm:p21}
$\chi(\Forb (\vec{C}_3, TT_3, P^+(2,1)) =3.$
\end{theorem}

This result will be a consequence of the following lemma.

\begin{lemma}\label{lem-p21}
Let $D$ be a $(\vec{C}_3, TT_3, P^+(2,1))$-free oriented graph. Then the following holds
\begin{enumerate}
\item[(1)] 
Every oriented odd hole in $D$ is directed.
\item[(2)] 
If a strong component of $D$ contains an odd hole, then it is an
initial strong component.
\item[(3)] 
If $D$ is strongly connected, then there is a stable set $S$ that
intersects every odd hole of $D$.
\end{enumerate}
\end{lemma}
We first observe that this lemma implies Theorem~\ref{thm:p21}.

\begin{proof}[Proof of Theorem~\ref{thm:p21} assuming Lemma~\ref{lem-p21}.]
Let $D_1, \dots , D_p$ be the initial strong components of $D$.  By
(3), for every $1\leq k\leq p$, there exists a stable set
$S_k\subseteq V(D_k)$ such that $D_k-S_k$ has no odd holes.  Now
$S=S_1\cup \dots \cup S_p$ is also a stable set because there is no
arc between two initial strong components, and by (1) and (2), $S$ is
a stable set that intersects every odd hole of $D$.  Since $D$ is
$(\vec{C}_3, TT_3)$-free, this implies that $D\setminus S$ is
bipartite, which concludes the proof.
\end{proof}

It remains to prove Lemma~\ref{lem-p21}.  

\begin{proof}[Proof of Lemma \ref{lem-p21}]\ To prove (1) it suffices to observe that every oriented odd
hole contains a directed path of size at least $2$.  Thus unless it is
directed it contains a $P^+(2,1)$.\\

Let us prove a claim that will imply (2) and (3).  A vertex $x\in
V(D)\setminus V(C)$ is a {\it $C$-twin of $v_i$} if $N^-(x)\cap
V(C)=\{v_{i-1}\}$ and $N^+(x)\cap V(C)=\{v_{i+1}\}$ (indices are taken
modulo $q$).

\begin{claim}\label{claim-odd-cycle}
Let $C=(v_1, \dots ,v_q, v_1)$ be a directed odd cycle in $D$, and let
$x$ be a vertex in $V(D)\setminus V(C)$.  Then:
\begin{itemize}
\item[(i)] 
$x$ is dominated by at most one vertex of $C$.
\item[(ii)] 
If there is $i$, such that $x$ dominates $v_{i+1}$ , then $x$ is a
$C$-twin of $v_i$.  
\item[(iii)] 
If $x\in \Reach^-(C)$, then $x$ is the $C$-twin of some $v_i$.
\end{itemize}
\end{claim}
\begin{subproof}
(i) Assume for a contradiction that $x$ is dominated by two vertices
in $C$.  Without loss of generality, we may assume that these two
vertices are $v_1$ and $v_i$ with $i < q/2$.  Then $(v_q,v_1, x,
v_{i})$ is an induced $P^+(2,1)$ in $D$, a contradiction.

\medskip 

(ii) Assume that $x\ra v_{i+1}$.  The path $(v_{i-1}, v_i, v_{i+1},
x)$ is a $P^+(2,1)$.  It is not induced, so $v_{i-1}\in N(x)$.  If
$x\ra v_{i-1}$, then with the same reasoning $v_{i-3} \in N(x)$.  We
can repeat this process as long as $x\ra v_{i+1- 2j}$.  However, this
process has to stop since $x$ is not adjacent to $v_{i+2}=
v_{i+1-2\lfloor q/2\rfloor}$.  Consequently, there exists $j$ such
that $x\la v_{i+1-2j}$ and $x\ra v_{i+1-2j'}$ for all $0\leq j' < j$.
But $j=1$ for otherwise $(v_{i+1-2j}, x, v_{i+1}, v_i)$ is an induced
$P^+(2,1)$.  Hence, $v_{i-1}\ra x$.

Now $x$ does not dominate any vertex $v_j\in V(C)\setminus
\{v_{i+1}\}$ for otherwise by the above reasoning both $v_{j-2}$ and
$v_{i-1}$ would dominate $x$, a contradiction to (i).  Therefore $x$
is a $C$-twin of $v_i$.

\medskip

(iii) Assume for a contradiction that $x\in \Reach^-(C)$ and $x$ is
not the $C$-twin of any $v_i$.  Let $P$ a be a shortest dipath from
$x$ to $C$.  Such a dipath exists because $x\in \Reach^-(C)$, and by
(ii), $P$ has length at least $2$.  Let $v_{i+1}$ be the terminal
vertex of $P$, $u$ its in-neighbour in $P$ and $t$ the in-neighbour of
$u$ in $P$.  The path $(t,u, v_{i+1}, v_i)$ is a $P^+(2,1)$, which is
not induced, so $t$ and $v_i$ are adjacent.  But $t$ does not dominate
$v_i$ since $P$ is a shortest dipath from $x$ to $C$, so $v_i\ra t$.

Since $u$ dominates $v_{i+1}$, we obtain that $u$ is a $C$-twin of
$v_i$ by (ii).  Therefore $C'=(v_1, \dots , v_{i-1},$ $u, v_{i+1},
\dots , v_q, v_1)$ is also a directed odd cycle.  By (ii), $t$ is a
$C'$-twin of $v_{i-1}$.  In particular, $v_{i-2}\ra t$.  This gives a
contradiction to (i) as $t$ is dominated by $v_{i-2}$ and $v_i$.
\end{subproof}

\noindent (2) now clearly follows from
Claim~\ref{claim-odd-cycle}~(iii).\\

\noindent (3) Suppose that $D$ is strongly connected.  If $D$
contains no oriented odd hole, then the result holds with
$S=\emptyset$.  If $D$ contains an odd hole $C=(v_1, \dots, v_q,
v_1)$, then it is directed by (1) and by Claim~\ref{claim-odd-cycle},
every vertex of $D$ is the $C$-twin of some $v_i$.  For $1\leq i\leq
q$, let $T_i$ be the set $C$-twins of $v_i$ plus $v_i$.  Observe that
if $xy\in A(D)$ with $x\in T_i$ and $y\in T_j$, then $|i-j| =1 \mod
q$, for otherwise $(v_{i-1},x,y,v_{j-1})$ would be an induced
$P^+(2,1)$.  It follows that $D-T_1$ has no odd cycles, and $T_1$ is
a stable set because all vertices in $T_1$ are in $N^-(v_2)$.  Thus
$T_1$ is our desired $S$.
\end{proof}

\begin{remark}
Wang and Wang~\cite{WaWa09} study a class of digraphs that contains
$\Forb (\vec{C}_3, TT_3, P^+(2,1))$.  A digraph is {\it arc-locally
in-semicomplete} if for any pair of adjacent vertices $x$, $y$, every
in-neighbour of $x$ and every in-neighbour of $y$ are either adjacent
or the same vertex.  Observe that the oriented graphs of $\Forb
(\vec{C}_3, TT_3, P^+(2,1))$ are {\it arc-locally in-semicomplete}.
In particular, \cite{WaWa09} characterizes strong arc-locally
in-semicomplete digraphs.  This characterization implies that every
strong oriented graph in $\Forb(P^+(2,1))$ is either a bipartite
tournament (i.e. the orientation of a complete bipartite graph) or an
extension of a directed cycle.  This directly implies
Lemma~\ref{lem-p21}~(3).

%
%
%A {\it $T$-digraph} is a digraph having a partition $(v_1\}, V_2, V_3, V_4)$ of its vertex set such that both $V_2$ and $V_3$ are non-empty stable sets, $D[V_4]$ is a possibly empty semicomplete digraph and $V_2\ra v_1$, $v_1 \mapsto V_3$, $V_3\cup V_4\mapsto V_2$, $V_4\mapsto V_3$ and $D[\{v_1\}\cup V_4]$ is semi complete.
%\begin{theorem}[Wang and Wang~\cite{WaWa09}]
%Let $D$ be a strong digraph in $\Forb(P^+(2,1))$, then D is either semicomplete, semicomplete bipartite, an
%extended cycle or a $T$-digraph.
%\end{theorem}
%
%
\end{remark}

\subsubsection{Forbidding $TT_3$ and $P^+(2,1)$.}
%%%%%%%%%%%%%%%%%%%%%%%%%%%%%%%%

We shall now prove that $\chi(\Forb(TT_3,P^+(2,1))) = 4$. % and that this result is tight. 
 Here is a short sketch of the proof.  We first
describe precisely the structure of a strong $(TT_3,P^+(2,1))$-free
oriented graph that contains an odd hole (see Lemma~\ref{oddholes}).
This permits us to colour such oriented graphs, more precisely we
distinguish between two cases, if the oriented graph contains an odd
hole of length $7$ or more, then it is $3$-colourable; if it contains
an odd hole of length $5$, then it is $4$-colourable.  We also give a
tight example in the second case (see Lemmas~\ref{7hole}
and~\ref{5hole}).  Finally we show how to $4$-colour any
$(TT_3,P^+(2,1))$-free oriented graph (Theorem~\ref{TT3P+(2,1)}).

\medskip

\begin{lemma}\label{oddholes} 
Let $D$ be a digraph in $\Forb(TT_3, P^+(2,1))$, and let $H=(v_1,
\ldots , v_{2k+1},v_1)$, $k\geq 2$, be an odd hole in $D$.  Then:
  \begin{itemize}
  \item[(i)] $H$ is directed.
  \item[(ii)] If $u\in \Reach^-(H)\setminus V(H)$, then $u$ is adjacent to some vertex of $H$.
  \item[(iii)] If $v$ dominates a vertex in $V(H)$, then either there is an index $i$ such that $vv_{i},v_{i-2}v$ are the only two arcs between $v$ and $V(H)$ or $k=2$ and there are exactly three arcs between $V(H)$ and $v$ and these are either $vv_{i},v_{i-2}v,vv_{i+2}$ or $vv_{i},v_{i-2}v, v_{i+1}v$ for some $i\in \{1, \dots , 5\}$.
    \item[(iv)] If $N^+(v)\cap V(H)=\emptyset$ but  $N^-(v)\cap V(H)\neq \emptyset$, then $|N^-(v)\cap V(H)|=1$.
    \item[(v)] $H$ is contained in an initial strong component of $D$.
  \end{itemize}
\end{lemma}

\begin{proof} 
In all this proof, indices of the $v_i$ are modulo $2k+1$.

(i) Every oriented odd hole contains a directed path of size at least
2.  Thus, unless it is directed, it contains a $P^+(2,1)$.

\medskip

(ii) Let $u$ be a vertex in $\Reach^-(H)\setminus V(H)$ outside $H$.
Let $P=(x_0, x_1\ldots , x_q)$ be a shortest $(u,V(H))$-dipath.
(Hence $u=x_0$).  If $q=1$ there is nothing to prove, so assume $q\geq
2$.  We may assume, by relabelling $V(H)$ if necessary, that
$x_q=v_{2k+1}$.  As $D$ is $TT_3$-free, the vertices $x_{q-1}$ and
$v_{2k}$ are not adjacent.  Consequently, as $D$ is $P^+(2,1)$-free,
$x_{q-2}$ must be adjacent to either $v_{2k+1}$ or to $v_{2k}$.  By
the minimality of $P$ the arc will enter $x_{q-2}$ in both cases.
Thus, since $D$ is $TT_3$-free, $D$ contains exactly one of those
arcs.  If $q=2$, we are done since $u$ is adjacent to a vertex of $H$
so suppose $q\geq 3$.  If $v_{2k}\ra x_{q-2}$ (resp.  $v_{2k+1}\ra
x_{q-2}$), then since $D$ is $(TT_3, P^+(2,1))$-free, the vertices
$v_{2k-1}$ (resp.  $v_{2k}$) and $x_{q-3}$ are adjacent, so, by
minimality of $P$, $v_{2k-1}\ra x_{q-3}$ (resp.  $v_{2k}\ra x_{q-3}$).
And so on by induction, one proves that there is an arc from $H$ to
$x_{q-4}, x_{q-5},$ etc until we get an arc from $H$ to $u$.  This
proves (ii).

\medskip

(iii) Let $v$ be a vertex in $V(D)\setminus V(H)$ that dominates a
vertex, say $v_i$, in $H$.  Moreover, without loss of generality, we
may assume that $vv_{i-2}$ is not an arc.  Indeed if $v$ dominates
$v_{i-2}$ for all $i$, then $v$ would dominate all vertices of $H$ and
$D$ would contain a $TT_3$.

Since $D$ is $TT_3$-free, then $v$ and $v_{i-1}$ are not adjacent.
Now there can be no arc $v_jv$ with $j\not\in \{i-2,i, i+1\}$ for
otherwise $(v_j,v,v_i, v_{i-1})$ would be an induced $P^+(2,1)$.
Furthermore, since $D$ has no induced $P^+(2,1)$, there is an arc
between $v$ and $v_{i-2}$.  By our assumption, this arc is
$v_{i-2}v_i$.  Now, there can be no arc $vv_j$ with $j\not\in
\{i-3,i\}$ for otherwise $(v_{i-2},v,v_j, v_{j-1})$ would be an
induced $P^+(2,1)$ or $D\langle\{v, v_{j-1}, v_j\}\rangle$ would be a
$TT_3$.  Consequently, in addition to $vv_i$ and $v_{i-2}v_i$, the
only possible arcs between $v$ and $H$ are $v_{i+1}v$ and $vv_{i-3}$.
If $k\geq 3$, then $vv_{i-3} \notin A(D)$ for otherwise $(v_{i-5},
v_{i-4}, v_{i-3}, v)$ is an induced $P^+(2,1)$, and $v_{i+1}v \notin
A(D)$, for otherwise $(v_{i-3}, v_{i-2}, v, v_{i+1})$ is an induced
$P^+(2,1)$.

If $k=2$, then $i-3=i+2$.  Both $v_{i+1}v$ and $vv_{i+2}$, cannot be
arcs for otherwise $\{v, v_{i+1}, v_{i+2}\}$ induces a $TT_3$.  This
completes the proof of (iii).

\medskip

(iv) Assume for a contradiction that $N^+(v)\cap V(H)=\emptyset$ and
$|N^-(v)\cap V(H)|\geq 2$.  There are distinct induces $i$ and $j$
such that $v_iv$ and $v_jv$ are arcs.  Observe that $i\notin \{j-1,
j+1\}$ because $D$ has no $TT_3$, and $v_{j-1}$ and $v$ are not
adjacent because $N^+(v)\cap V(H)=\emptyset$ and $D$ has no $TT_3$.
%Thus $(v_{j-1},v_j,v,v_i)$ is an induced $P^+(2,1)$, contradiction.
If $|j-2|\neq 2$ then $(v_{j-1},v_j,v,v_i)$ is an induced $P^+(2,1)$ and if $i=j-2$, then $(v_{i-1},v_i,v,v_j)$  is an induced $P^+(2,1)$, so we obtain the desired contradiction.
\medskip

(v) Suppose for a contradiction that $H$ is contained in a strong
component $C$ that is not initial.  Then there is a vertex $u\in
\Reach^-(H)\setminus V(C)$ such that $u$ belongs to an initial component.  By (ii), $u$ is adjacent to a vertex in
$H$.  If $u$ dominates a vertex in $H$, then by (iii) it is also
dominated by a vertex of $H$.  Hence in any case, $u$ is dominated by
a vertex of $H$.  But this implies that $u\in C$, a contradiction.
 \end{proof}

\begin{lemma}\label{7hole}
Let $D$ be a strong digraph in $\Forb(TT_3, P^+(2,1))$.  If $D$
contains an odd hole $H$ with at least 7 vertices, then $D$ is an
extension of $H$.  In particular $\chi(D)=3$.
\end{lemma}

\begin{proof} 
Let $H=(v_1, \dots , v_{2k+1},v_1)$, $k\geq 3$ be an odd hole in $D$.
By Lemma \ref{oddholes}~(i)--(ii), $H$ is directed and every vertex of
$V(D)\setminus V(H)$ is adjacent to $V(H)$.  Suppose $D$ is not an
extension of $H$.  Then by Lemma \ref{oddholes}~(iii)--(iv) there is a
vertex $x_1$ such that $N^+(x_1)\cap V(H)=\emptyset$ and
$|N^-(x_1)\cap V(H)|= 1$.  Let $v_j$ be the vertex of $N^-(x_1)\cap
V(H)$.  As $D$ is strong there exists a $(x_1, H)$-dipath.  Let
$P=(x_1,x_2, \dots ,x_t,v_i)$ be a shortest such dipath.  Then by
minimality of $P$, $x_t$ is the only vertex of $P-\{v_i\}$ that has an
arc to $V(H)$.  By Lemma~\ref{oddholes}~(iii), $v_{i-2}x_t,x_tv_i$ are the
only arcs between $x_t$ and $V(H)$.  Now, $x_{t-1}$ must be adjacent to
$v_{i-3}$, otherwise $(v_{i-3},v_{i-2},x_t,x_{t-1})$ is an induced
$P^+(2,1)$. As $x_{t-1}$ has no arc to $V(H)$ we have that $v_{i-3}x_{t-1}$ is an arc and by Lemma~\ref{oddholes}~(iv) this is the only arc between $x_{t-1}$ and $V(H)$, implying that
$(x_{t-1},x_t,v_i,v_{i-1})$ is an induced $P^+(2,1)$, a contradiction.
\end{proof}

\begin{lemma}\label{5hole}
 Let $D$ be a strong $(TT_3, P^+(2,1))$-free oriented graph.  If $D$
 contains a 5-hole, then $\chi(D) \le 4$.
%We also describe such a graph with chromatic number $4$.
\end{lemma}

\begin{proof}
Let $H=(v_1,v_2,v_3,v_4,v_5,v_1)$ be a 5-hole in $D$.  For $i=1,\dots,
5$, define (subscripts are taken modulo $5$ all along the proof):
\begin{itemize}
\item 
$A_i = \{ v \in D-H : \text{$v\la v_{i-1}$, and $v \ra \{v_{i+1} ,v_{i+3}\}$} \}$.
\item 
$B_i = \{ v \in D-H : \text{$v\la \{v_{i-1},v_{i+2}\}$, and $v \ra v_{i+1}$} \}$.
\item 
$C_i = \{ v \in D-H : \text{$v \ra v_{i+1}$ and $v\la v_{i-1}$ }\}$.
\item 
$X_i=A_i\cup B_i \cup C_i$.
\end{itemize}
By Lemma~\ref{oddholes}~(ii)-(iv), the sets $X_1, \dots, X_5$ are a 
partition of the set $V(D)\setminus V(H)$.  Moreover, since $D$ is
$TT_3$-free, we have:
\begin{claim}\label{noArcs}
For $i=1, \dots, 5$, $X_i$ is a stable set, and there is no arc
between $X_i$ and $B_{i+2}$ or between $X_i$ and $A_{i+3}$.
\end{claim}

%%for every vertex $x \in V(D)$ in $D$ there is $i$ such that $v_{i-1}$ dominates $x$ anf $v_{i+1}$ is dominated by $x$.  Therefore $V(D)$ can be partitioned into five sets $X_1,\ldots ,X_5$ such that 5for all $i\in \{1, \dots , 5\}$, every vertex $x\in X_i$ is dominated by $v_{i-1}$. 

Let $\pi$ be the colouring of $D$ defined as follows (see Figure~\ref{figCol}).
\begin{itemize}
\item $\pi(v_1)=  1$, $\pi(v_2)=\pi(v_5) = 2$, $\pi(v_3)=3$ and $\pi(v_4)=4$;
\item $\pi(x)=1$ for all $x\in X_1\cup A_4\cup B_2$;
\item $\pi(x) = 2$ for all $x\in A_5 \cup C_5$;
\item $\pi(x) = 3$ for all $x\in X_3\cup B_5$;
\item $\pi(x)=4$ for all $x\in A_2 \cup B_4 \cup C_4$.
\end{itemize}
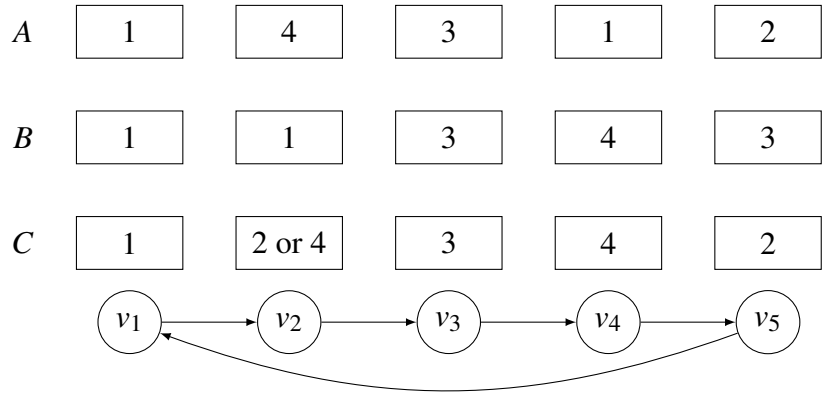
\begin{figure}[hbtp]
\centering
 \begin{tikzpicture}[scale=0.7]
\node[draw,circle] (1) at (0,0) {$v_1$};
\node[draw,circle] (2) at (3,0) {$v_2$};
\node[draw,circle] (3) at (6,0) {$v_3$};
\node[draw,circle] (4) at (9,0) {$v_4$};
\node[draw,circle] (5) at (12,0) {$v_5$};
%\draw  [fill=black]  (0, 0) circle(2.0pt);
%\draw  [fill=black]  (3, 0) circle(2.0pt);
%\draw  [fill=black] (6, 0) circle(2.0pt);
%\draw  [fill=black] (9, 0) circle(2.0pt);
%\draw [fill=black]  (12, 0) circle(2.0pt);

\draw[->,>=latex] (1) -- (2);
\draw[->,>=latex] (2) -- (3);
\draw[->,>=latex] (3) -- (4);
\draw[->,>=latex] (4) -- (5);
\draw[->,>=latex] (5) to[bend left=20] (1); 

%1
\draw(-1,1) rectangle (1,2);
\draw(-1,3) rectangle (1,4);
\draw(-1,5) rectangle (1,6); 
\draw (0,5.5) node {$1$}; %A
\draw (0,3.5) node {$1$}; %B
\draw (0,1.5) node {$1$}; %C

%2
\begin{scope}[xshift=3cm]
\draw(-1,1) rectangle (1,2);
\draw(-1,3) rectangle (1,4);
\draw(-1,5) rectangle (1,6);
\draw (0,5.5) node {$4$}; %A
\draw (0,3.5) node {$1$}; %B
\draw (0,1.5) node {$2$ or $4$}; %B
\end{scope}

%3
\begin{scope}[xshift=6cm]
\draw(-1,1) rectangle (1,2);
\draw(-1,3) rectangle (1,4);
\draw(-1,5) rectangle (1,6);
\draw (0,5.5) node {$3$}; %A
\draw (0,3.5) node {$3$}; %B
\draw (0,1.5) node {$3$}; %C
\end{scope}

%4
\begin{scope}[xshift=9cm]
\draw(-1,1) rectangle (1,2);
\draw(-1,3) rectangle (1,4);
\draw(-1,5) rectangle (1,6);
\draw (0,5.5) node {$1$}; %A
\draw (0,3.5) node {$4$}; %B
\draw (0,1.5) node {$4$}; %C
\end{scope}

%5
\begin{scope}[xshift=12cm]
\draw(-1,1) rectangle (1,2);
\draw(-1,3) rectangle (1,4);
\draw(-1,5) rectangle (1,6);
\draw (0,5.5) node {$2$}; %A
\draw (0,3.5) node {$3$}; %B
\draw (0,1.5) node {$2$}; %C
\end{scope}

\draw (-2,1.5) node {$C$};
\draw (-2,3.5) node {$B$};
\draw (-2,5.5) node {$A$};
\end{tikzpicture}
  \caption{The colouring $\pi$ of $D-C_2$}
  \label{figCol}
\end{figure}

By Claim~\ref{noArcs}, $\pi$ is a proper colouring of $D-C_2$.
%Let us now prove that $\pi$ is a proper colouring of $D-C_2$. Since for each $i$ the vertex $v_{i-1}$ dominates the set $X_i$ and the set $B_{i+2}$,  there are no arcs between $X_i$ and $B_{i+2}$   because $D$ is $TT_3$-free. Similarly there is no arc between  $X_i$ and $A_{i-2}$ since both sets dominate the vertex $v_{i+1}$. Hence, the function is for the moment a proper colouring. It only remains to define $\pi$ for the set $C_2$. 

For any $v \in C_2$, set $\pi(v)=4$ if $v$ has a neighbour in $C_5$,
and $\pi(v)=2$ otherwise.  We shall prove that the function $\pi$ is a
proper colouring of $D$.  By Claim~\ref{noArcs}, if $v\in C_2$ has no
neighbour in $C_5$, then none of its neighbours is coloured $2$.  So
the only problem that might occur is if a vertex of $v \in C_2$
coloured with $4$ (and thus adjacent to a vertex $u \in C_5$) has a
neighbour with colour $4$, say $w$.
%Suppose $v \in C_2$ has a neighbour $w$ with $\pi(w)=4$. 
By Claim~\ref{noArcs}, $w \in C_4$ and since $D$ is $TT_3$-free, $vu$
and $wv$ are arcs of $D$.
%Since $v$ dominates $v_3$ and $v_3$ dominate $w$ then $w$ dominates $v$. By definition $v$ there exists a vertex $v' \in C_5$ which is a neighbour of $v$. Moreover, since $v'$ dominates $v_1$ and $v_1$ dominates $v$ then $v$ dominates $v'$. 
  
%We claim that $uw$ must be an arc of $D$.  
If $u$ and $w$ were non-adjacent, then  $(w,v,u,v_4)$ would be  an
induced $P^+(2,1)$.  So they are  adjacent and $u$ dominates $w$, since $(u,v,w)$ cannot induce a $TT_3$. But then $(v_2,v_3,w,u)$ is
an induced $P^+(2,1)$, a contradiction.  This proves that $\pi$ is a
proper colouring of $D$ and then $\chi(D) \leq 4$.
\end{proof}

We now describe a $(TT_3,P^+(2,1))$-free oriented graph with chromatic
number $4$.  Take two 5-holes $C_1= (v_1,v_2,v_3,v_4,v_5,v_1)$ and
$C_2=(u_1,u_2,u_3,u_4,u_5,u_1)$ and for each vertex $u_i$ we add the
arcs $v_{i-1}u_i, u_iv_{i+1}$ and $u_iv_{i+3}$ (see figure
\ref{figCol4}).  It is a routine exercise to check that this oriented
graph is indeed $(TT_3, P^+(2,1))$-free.  In any $3$-colouring of
$C_1$, there exists $i \in \{1, \dots , 5\}$ such that the vertices
$v_{i-1},v_{i+1},v_{i+3}$ have distinct colours, and thus no colour is
available for $u_i$.  So this graph is not $3$-colourable.
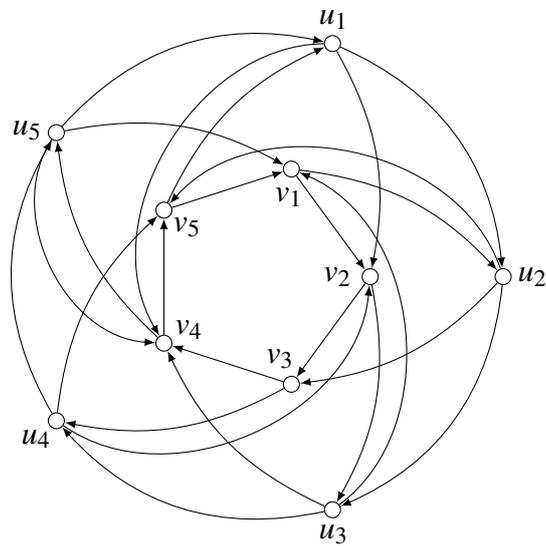
\begin{figure}[hbtp]
\centering
\begin{tikzpicture}[scale=0.5]
\vertex[](v2) at  (0:3) {};  \draw (0:2.8) node[left]{$v_2$};
\vertex[](v1) at  (72:3) {};   \draw (72:2.8) node[below]{$v_1$};
\vertex[](v5) at  (144:3) {};  \draw (152:2.8) node[right]{$v_5$};
\vertex[](v4) at  (216:3) {};  \draw (210:2.8) node[right]{$v_4$};
\vertex[](v3) at  (288:3) {};  \draw (283:2.8) node[above]{$v_3$};

\vertex[](u2) at  (0:6.5) {};  \draw (0:6.6) node[right]{$u_2$};
\vertex[](u1) at  (72:6.5) {};  \draw (72:6.6) node[above]{$u_1$};
\vertex[](u5) at  (144:6.5) {}; \draw (144:6.6) node[left]{$u_5$};
\vertex[](u4) at  (216:6.5) {}; \draw (213:6.9) node[below]{$u_4$};
\vertex[](u3) at  (288:6.5) {}; \draw (288:6.6) node[below]{$u_3$};

\draw[->,>=latex] (v1) to  (v2); 
\draw[->,>=latex] (v2) to  (v3); 
\draw[->,>=latex] (v3) to  (v4); 
\draw[->,>=latex] (v4) to  (v5); 
\draw[->,>=latex] (v5) to  (v1); 

\draw[->,>=latex] (u1) to[bend left=30]   (u2); 
\draw[->,>=latex] (u2) to[bend left=30]   (u3); 
\draw[->,>=latex] (u3) to[bend left=30]   (u4); 
\draw[->,>=latex] (u4) to[bend left=30]   (u5); 
\draw[->,>=latex] (u5) to[bend left=30]   (u1); 

\draw[->,>=latex] (v5) to[bend left=20]  (u1); 
\draw[->,>=latex] (u1) to[bend left=20]  (v2);
\draw[->,>=latex] (u1) to[bend right=60]  (v4);

\draw[->,>=latex] (v1) to[bend left=20]  (u2); 
\draw[->,>=latex] (u2) to[bend left=20]  (v3);
\draw[->,>=latex] (u2) to[bend right=60](v5);

\draw[->,>=latex] (v2) to[bend left=20]  (u3); 
\draw[->,>=latex] (u3) to[bend left=20]  (v4);
\draw[->,>=latex] (u3) to[bend right=60](v1);

\draw[->,>=latex] (v3) to[bend left=20]  (u4); 
\draw[->,>=latex] (u4) to[bend left=20]  (v5);
\draw[->,>=latex] (u4) to[bend right=60](v2);

\draw[->,>=latex] (v4) to[bend left=20]  (u5); 
\draw[->,>=latex] (u5) to[bend left=20]  (v1);
\draw[->,>=latex] (u5) to[bend right=60](v4);
\end{tikzpicture}
  \caption{A $(TT_3,P^+(2,1))$-free oriented graph with chromatic number $4$.  }
  \label{figCol4}
\end{figure}

\begin{theorem}\label{TT3P+(2,1)}
$\chi(\Forb(TT_3, P^+(2,1))) = 4$.  More precisely, if
$D$ is a $(TT_3, P^+(2,1))$-free oriented graph, then the following
hold.
\begin{itemize}
\item $\chi(D) \le 4$; \item If $D$ contains an odd hole of length $7$
or more, then $\chi(D) =3$.
\end{itemize}

\end{theorem}

\begin{proof}
Let $D \in \Forb(TT_3, P^+(2,1))$ and assume $D$ is connected.  We may
assume that $D$ admits at least one initial strong component $K$ that
contains an odd hole, otherwise by Lemma~\ref{oddholes}~(v) $D$ is
odd hole-free and thus is $4$-colourable by Theorem~\ref{thm:CRST10}.

\begin{claim}
$K$ is the only initial strong component of $D$.
\end{claim}
\begin{subproof}
Assume $D$ contains another initial strong component $K'$.  Let
$P=(p_1,p_2, \dots ,p_k)$ be a shortest path from $K$ to $K'$, where
$p_1 \in K$ and $p_k \in K'$.  Note that since $K$ and $K'$ are
initial strong components $p_1\ra p_2$ and $p_k\ra p_{k-1}$.  As $K$ is strong and non-trivial there
exists a vertex $p_0$ in $V(K)\setminus \{p_1\}$ such that $p_0 \ra
p_1$.  Observe that by minimality of $P$, and since $D$ is
$TT_3$-free, $P'=(p_0,p_1, \dots ,p_k)$ is an induced path.  Moreover,
since $p_0 \ra p_1 \ra p_2$ and $p_{k-1} \la p_k$, necessarily $P'$
contains a $P^+(2,1)$, a contradiction.
\end{subproof}

Let $\dist(K,x)$ denote the distance from $K$ to $x$, that is the
length of a shortest dipath from $K$ to $x$ in $D$.  Note that
$\dist(K,x)$ is well-defined for every vertex $x\in V(D)$, because $K$
is the only initial strong component, so every vertex can be reached
from $K$.  Set $L_i=\{x: \dist(K,x)=i\}$ (in particular $L_0=K$).
Clearly, the $L_i$ partition $V(D)$.  If $j>i$, an arc from $L_j$ to
$L_i$ is called a \emph{backward arc}.

\begin{claim}
$D$ has no backward arcs.
\end{claim}

\begin{subproof}
Assume for contradiction that $uv$ is a backward arc from $L_j$ to
$L_i$ and assume it has been chosen with respect to the minimality of
$i$.  Observe that $i \ge 1$.  If $i \ge 2$, then there exists a
vertex $v_1 \in L_{i-1}$ and a vertex $v_2 \in L_{i-2}$ such that $v_2
\ra v_1 \ra v$ and thus $(v_2,v_1,v,u)$ is a $P^+(2,1)$ and it is
induced by minimality of $i$, a contradiction.  So we may assume that
$i=1$.  Let $v_1 \in L_0$ such that $v_1 \ra v$.  There exists a
vertex $v_2 \in L_0$ such that $v_2 \ra v_1$ and since $D$ is
$TT_3$-free, $v_2$ is not adjacent to $v$.  Hence $\{v_2,v_1,v,u\}$
induces a $P^+(2,1)$, a contradiction.
\end{subproof}

\begin{claim}
For any $i \ge 2$, $L_i$ is a stable set.
\end{claim}
\begin{subproof}
Let $i \ge 2$ and assume that $uv$ is an arc of $L_i$.  There exists
$v_1 \in L_{i-1}$ and $v_2 \in L_{i-2}$ such that $v_2 \ra v_1 \ra v$.
So $(v_2,v_1,v,u)$ is a $P^+(2,1)$ and it is induced since there is no
$TT_3$ nor backward arcs.
\end{subproof}

A \emph{directed bipartite graph} is an orientation of a connected
bipartite graph such that every vertex is either a source or a sink.

\begin{claim}\label{dibip}
$L_1$ is a  disjoint union of  directed bipartite graphs. 
\end{claim}
\begin{subproof}
Assume for contradiction that there exists $a,b,c \in L_1$ such that
$a \ra b \ra c$ (note that $ca$ might or might not be an arc).  We
distinguish between two cases.  \medskip

\noindent{\bf Case 1}: $c$ admits a neighbour $c_1 \in L_0$ such that
$c_1$ belong to an odd hole $H=(c_1, \dots, c_{2k+1},c_1)$ of $L_0$.
Since $(c_{2k+1},c_1,c,b)$ cannot be induced, $c_{2k+1}\ra b$ and
since $(c_{2k},c_{2k+1},b,a)$ cannot be induced, $c_{2k}\ra a$.
Recall that by Lemma~\ref{oddholes} (iv), a vertex in $L_1$ is
adjacent to at most one vertex in $H$.  Since $(a,b,c,c_1)$ cannot be
induced, we must have $c\ra a$.  But now $(c_1,c,a,c_{2k})$ is an induced
$P^+(2,1)$, a contradiction.  \medskip

\noindent{\bf Case 2}: no neighbour of $c$ in $L_0$ belongs to an odd
hole in $L_0$.  Let $c_1 \in L_0$ be a neighbour of $c$.  By
Lemma~\ref{7hole}, if $L_0$ contains an odd hole of length at least
$7$, then all vertices of $L_0$ belong to an odd hole.  So we may
assume that $L_0$ contains a $5$-hole, say
$H=(u_1,u_2,u_3,u_4,u_5,u_1)$.  By property~\ref{oddholes}~(iii), we
may assume without loss of generality that $u_2 \ra c_1 \ra u_4$ and
that exactly one of $u_5c_1$, $c_1u_1$ is an arc.  Recall again that
by Lemma~\ref{oddholes} (iv), a vertex in $L_1$ is adjacent to at most
one vertex in $H$.

Since $(u_2,c_1,c,b)$ cannot be induced, $u_2b$ is an arc.  Since
$(u_1,u_2,b,a)$ cannot be induced, $u_1a$ is an arc.  Since
$(a,b,c,c_1)$ cannot be induced and $c_1a$ is not an arc by Lemma~\ref{oddholes} (iv), $c$ and $a$ are adjacent and we have $c\ra a$. 
But now 
$(u_5,u_1,a,c)$ is an induced $P^+(2,1)$ (it is indeed
induced because $c$ has no neighbour in $H$), a contradiction.
\end{subproof}

We may now assume that $L_1$ consists of $t$ directed bipartite graphs
$(A_1,B_1), \dots, (A_t,B_t)$ such that all arcs of $L_1$ are from
$A_i$ to $B_i$.

\begin{claim}\label{diBip}
Let $1 \le i \le t$ and let $u,v \in A_i$.  Then $u$ and $v$ have the
same neighbourhood in $L_0$ and the graph induced by $N_{L_0}(B_i)$
and $N_{L_0}(A_i)$ is a complete bipartite graph.
\end{claim}
\begin{subproof}
Assume for a contradiction that there exists a vertex $u' \in L_0$
such that $u'u$ is an arc but $u'v$ is not.  As $(A_i,B_i)$ is connected, we may assume without
loss of generality that $u$ and $v$ have a common neighbour in $B_i$,
say $w$.  Then $(u',u,w,v)$ induces a $P^+(2,1)$, a contradiction.

Let $w \in B_i$ and $w'$ be a neighbour of $w$ in $L_0$.  Let $u\in
A_i$ be a neighbour of $w$.  Let $u' \in N_{L_0}(A_i)$.  Since $u'$
dominates all vertices of $A_i$, $u'$ dominates $u$ and thus $u' \neq
w$, otherwise $(u',u,w)$ is a $TT_3$, and $u'$ is adjacent to $w'$,
otherwise $(u',u,w,w')$ induce a $P^+(2,1)$.
\end{subproof}

\begin{claim}\label{ngbBelow}
Let $i \ge 2$ and let $u \in L_i$.  Then the neighbours of $u$ in
$L_{i-1}$ have the same neighbourhood in $L_{i-2}$.
\end{claim}
\begin{subproof}
Let $v,w$ be two neighbours of $u$ in $L_{i-1}$.  Since there is no
backward arcs, $vu$ and $wu$ are arcs.  If some $z\in L_{i-2}$ was adjacent to precisely one of $v,w$, say $v$, then $(z,v,u,w)$ would induce  a $P^+(2,1)$.
Hence $v$ and $w$ share the
same in-neighbourhood, which implies the claim.
\end{subproof}

We are now going to explain how a   $k$-colouring of $L_0$ (where $k=3$ or $4$), can be extended to the rest of the graph. 
So assume that $L_0$ is coloured with colours from $\{1,2, \dots, k\}$. 

We start by colouring $L_1$.  Let $1 \le i \le t$ and let $I \subseteq
\{1, \dots, k\}$ be the set of colours used to colour $N_{L_0}(A_i)$.
Since $N_{L_0}(B_i)$ is complete to $N_{L_0}(A_i)$, $I\neq \{1, \dots,
k\}$ and only colours from $\{1, \dots, k\}-I$ are used to colour
$N_{L_0}(B)$.  So we can colour the vertices of $A_i$ with a colour
from $\{1, \dots, k\}-I$ and the vertices in $B_i$ with a colour from
$I$.  Hence we can colour all vertices of $L_1$.  Moreover assume we
are doing so in such a way that two vertices of $L_1$ that are sharing
the same neighbourhood in $L_0$ are coloured with the same colour.

Now we colour the rest of the graph layer by layer.  Assume that all
layer below $L_i$ ($i \ge 2$) have already been coloured in such a way
that two vertices in the same layer that have the same neighbour in
the layer below are coloured with the same colour.  Then, by
Claim~\ref{ngbBelow}, each vertex in $L_i$ see a single colour in
$L_{i-1}$, so it is easy to extend the colouring.
\end{proof}

\subsection{Forbidding several orientations of $P_4$}
%%%%%%%%%%%%%%%%%%%%%%%%%%%
%\begin{theorem}\label{prop:saufp3}
%$\Forb (P^+(3), P^+(2,1), P^-(2,1))$ is $\chi$-bounded.
%\end{theorem}
%\begin{proof}
%Let $T$ be the tree with vertex set $\{a, b_1, b_2, b_3, c_1, c_2\}$ and edge set $\{a b_1, a b_2, a b_3, b_1c_1, b_2c_2\}$.
%It is easy to check that any orientation of $T$ contains one of $\{P^+(3), P^+(2,1), P^-(2,1)\}$ as (induced) subdigraph.
%Therefore $\Forb (P^+(3), P^+(2,1), P^-(2,1))\subseteq \Forb(\Or(T))$.
%Now Kierstead and Penrice~\cite{KiPe94} proved that $\Forb(\Or(T))$ is $\chi$-bounded, so $\Forb (P^+(3), P^+(2,1), P^-(2,1))$ is $\chi$-bounded.
% \end{proof}

Observe that, by directional duality, $\Forb(P^+(3), P^+(2,1))
=\Forb(P^+(3), P^-(2,1))$.

\begin{proposition}\label{prop:p3-p21}
An oriented graph in $\Forb(P^+(3), P^+(2,1))$ or $\Forb(P^+(3),
P^+(1,1,1))$ contains no odd hole.
\end{proposition}
\begin{proof}
Let $D$ be a $(P^+(3), P^+(2,1))$-free oriented graph.  Assume for a
contradiction, that it contains an odd hole $C=(v_1, \dots ,v_p,
v_1)$.  Necessarily, $C$ contains two consecutive edges that are
oriented in the same direction.  Without loss of generality, $v_1\ra
v_2 \ra v_3$.  Now $(v_1, v_2, v_3, v_4)$ is either a $P^+(3)$ or a
$P^+(2,1)$, a contradiction.

Let $D$ be a $(P^+(3), P^+(1,1,1))$-free oriented graph.  Assume for a
contradiction, that it contains an odd hole $C=(v_1, \dots ,v_p,
v_1)$.  Necessarily, $C$ contains two edges at distance $1$ that are
oriented in the same direction.  Without loss of generality, $v_1\ra
v_2$ and $v_3 \ra v_4$.  Now $(v_1, v_2, v_3, v_4)$ is either a
$P^+(3)$ or a $P^+(1,1,1)$, a contradiction.
\end{proof}

A recent and difficult paper of Seymour and Scott (see \cite{ScSe})
proves that the class of odd-hole-free graphs is $\chi$-bounded, which
directly yields the following results.

\begin{corollary}\label{cor:2P4}
$\Forb(P^+(3), P^+(2,1))$, $\Forb(P^+(3), P^-(2,1))$, and
$\Forb(P^+(3), P^+(1,1,1))$ are $\chi$-bounded.
\end{corollary}

A natural question is to ask for the values (or nice bounds) of
$\chi(\Forb (\Or(K_k), P^+(3), P^+(2,1)))$ and $\chi(\Forb (\Or(K_k),
P^+(3), P^+(1,1,1)))$ for every $k\geq 3$.  A graph with no odd hole
nor clique of size $3$ contains no odd cycle and thus is bipartite.
Thus
\begin{proposition}
$\chi(\Forb (\vec{C}_3, TT_3, P^+(3), P^+(2,1)))=\chi(\Forb
(\vec{C}_3, TT_3, P^+(3), P^+(1,1,1))) = 2$.
\end{proposition}

One can also easily prove the following proposition.
\begin{proposition}
$$\chi(\Forb (\vec{C}_3, TT_3, P^+(2,1), P^+(1,1,1))) =3.$$  
\end{proposition}
%\begin{proof}
%Assume for a contradiction that a digraph in $\Forb (\vec{C}_3, TT_3, P^+(2,1), P^+(1,1,1))$ has chromatic number greater than  $3$.
%Then there is a $3$-critical digraph $D$ in  $\Forb (\vec{C}_3, TT_3, P^+(2,1), P^+(1,1,1))$.

%\begin{proposition}\label{prop:s0j}
%	  $\univ(S_{0,j}) = 2j$, for all $j \geq 1$. 
%	\end{proposition}
%{\bf regarder comment rediger}
%By Proposition~\ref{prop:s0j}, there is a vertex $v$ in $D$ having two distinct in-neighbours $u_1$ and $u_2$ in $D$.
%By Proposition~\ref{prop:distinct}, there a vertex $w$ in $N(u_2)\neq N(u_1)$.
%If $u_2\ra w$, then $(u_1, v, u_2, w)$ is an induced  $P^+(1,1,1)$.
%If $u_2\la w$, then $(w, u_2, v, u_1)$ is an induced $P^+(2,1)$.
 %In both cases, we obtained a contradiction. Hence $\chi(\Forb (\vec{C}_3, TT_3, P^+(2,1), P^+(1,1,1))) \leq 3$.
%Now every directed cycle is in  $\Forb (\vec{C}_3, TT_3, P^+(2,1), P^+(1,1,1))$ and has chromatic number $3$.
%\end{proof}
%It is not known whether $\Forb(P^+(2,1),P^-(2,1),P^+(1,1,1))$ is $\chi$-bounded.

\noindent This proposition also derives directly from
Theorem~\ref{thm:p21} and the fact that directed odd cycles are in
$\Forb (\vec{C}_3, TT_3, P^+(2,1), P^+(1,1,1))$.

\section{Concluding Remarks}

Let us conclude by discussing the remaining open cases.  Conjecture
\ref{conj:star} about stars is still widely open, the next case to
study being $\Forb(\Or(K_4),S_{k,k})$.  About oriented paths, note
that since $\Forb(P^+(3))$ and $\Forb(P^+(1,1,1))$ are not
$\chi$-bounded, the only open cases for orientations of $P_k$ that
would be $\chi$-bounding are paths of the type $P^+(2,2,\ldots,2)$ or
$P^+(1,2,2,\ldots,2)$, or $P^+(1,2,2,\ldots,2,1)$ (following our
notations).  In fact for trees in general, most orientations will
contain either $P^+(3)$ and $P^+(1,1,1)$ and hence when forbidden will
define classes that are not $\chi$-bounded.

Recall that Conjecture \ref{gyarfas-sumner} states that for every tree
$T$, the class of $T$-free graphs is $\chi$-bounded.  A stronger
conjecture could be the following : for every tree $T$, there exists
one orientation $\vec{T}$ of $T$ such that the class of graphs that
admit a $\vec{T}$-free orientation is $\chi$-bounded.  This is false
for many trees, as shown below.
\begin{proposition}
There exists a tree $T$ such that for every orientation $\vec{T}$ of
$T$, $\Forb(\vec{T})$ is not $\chi$-bounded.
\end{proposition}
\begin{proof}
To construct $T$, start with an induced path on four vertices
$\{v_1,v_2,v_3,v_4\}$ and add vertices $\{w_1,w_2,w_3,w_4\}$ such that
$N(w_i)=\{v_i\}$.  It is easy to see that every orientation of this
tree contains either a $P^+(3)$ or $P^+(1,1,1)$.  Therefore
$\Forb(\vec{T})$ contains either $\Forb(P^+(3))$ or
$\Forb(P^+(1,1,1))$ which are both not $\chi$-bounded.
\end{proof}
Of course any tree that contains this tree $T$ will also satisfy the
theorem.  Up to our knowledge, Gy\'arf\'as-Summner conjecture
(Conjecture~\ref{gyarfas-sumner}) is not known to be true for these
trees, so they could be natural candidates for counterexamples.

\section*{Acknowledgement}
%%%%%%%%%%%%%%
The authors would like to thank St\'ephan Thomass\'e for stimulating discussions.

\end{document}